\documentclass[11pt]{article}

\usepackage[letterpaper,margin=1in]{geometry}
\usepackage{amsmath,amssymb,amsthm,mathtools}
\usepackage{microtype}
\usepackage{graphicx}
\usepackage{placeins}
\numberwithin{equation}{section}
\usepackage{xcolor}
\usepackage{booktabs}
\usepackage{tikz}
\usetikzlibrary{arrows.meta}
\usepackage{algorithm}
\usepackage[noend]{algpseudocode}
\usepackage[
  colorlinks=true,
  linkcolor=blue!55!black,
  citecolor=blue!55!black,
  urlcolor=blue!55!black
]{hyperref}
\usepackage{authblk}

\newtheorem{theorem}{Theorem}[section]
\newtheorem{lemma}[theorem]{Lemma}
\newtheorem{proposition}[theorem]{Proposition}
\newtheorem{corollary}[theorem]{Corollary}

\theoremstyle{definition}
\newtheorem{definition}[theorem]{Definition}

\newtheorem{remark}[theorem]{Remark}

\newcommand{\ALG}{\operatorname{ALG}}
\newcommand{\OPT}{\operatorname{OPT}}
\newcommand{\cost}{\operatorname{cost}}
\newcommand{\len}{\operatorname{len}}
\newcommand{\dist}{\operatorname{dist}}

\newcommand{\ST}{\operatorname{ST}}
\newcommand{\MST}{\operatorname{MST}}

\title{Online Service with Per-Batch Maximum Delays}
\author{
Tianhang Lu\thanks{The first two authors contribute equally.},
Runtian Ren,
Shengcai Liu,
Ke Tang
}
\affil{
Guangdong Provincial Key Laboratory of Brain-Inspired Intelligent Computation,\\
Department of Computer Science and Engineering,\\
Southern University of Science and Technology, Shenzhen 518055, China\\
liusc3@sustech.edu.cn
}
\date{}

\begin{document}

\pagenumbering{roman}
\maketitle
\thispagestyle{empty}

\begin{abstract}
We study online service with per-batch maximum delays on a finite metric space.
Requests arrive over time, and a server serves each batch by following a walk visiting its requests' locations.
Movement is instantaneous.
The objective is to minimize the total distance traveled by the server plus the sum, over all batches, of the maximum waiting time in each batch.

We give deterministic polynomial-time algorithms with competitive ratios $10$ on finite lines, $12$ on weighted trees, and $20$ on arbitrary finite metrics.
Unlike additive delay, maximum delay allows younger requests to join a batch without increasing its delay charge.
Our algorithms coordinate service using a common waiting clock and a threshold determined by the metric.
The analysis uses certificates for groups of requests to compare this policy with any offline schedule, even one that serves different batches and follows different walks.
This comparison shows that coordinating service in this way loses only a constant factor.

The bounds hold whether the server must serve requests it visits or may leave them unserved; the two rules also have equal offline optimal costs.
Offline optimization is polynomial on lines and explicitly represented weighted trees, and NP hard on general metrics.
When arrivals are announced but locations are hidden until visited, we give an $84$-competitive algorithm on a known finite line under either rule.
\end{abstract}

\clearpage
\pagenumbering{arabic}
\section{Introduction}
\label{sec:introduction}

Online service with delay studies a basic tradeoff in dynamic routing: a
mobile server can wait and serve nearby requests together, but waiting
increases their delays.  The classical problem was introduced by Azar, Ganesh,
Ge, and Panigrahi~\cite{azar2017osd}.  Requests arrive over time at points of
a metric space, and the objective is the distance traveled by the server plus
the sum of the requests' individual delay costs. Unlike the classical
$k$-server problem, requests need not be served immediately.
The decision of when to wait also arises in same-day courier services and
dynamic vehicle routing, where experiments have studied its effects on
route length and the ability to serve new customers
\cite{mitrovic2004waiting,branke2005waiting}. Online service also belongs to a broader family of metric
optimization problems with delay or deadlines, including facility location
and network-design problems~\cite{azar2019framework,azar2020beyond}.

The sum of individual delays is natural when total waiting is the relevant
quantity. When a batch's delay charge is determined by its oldest request,
however, adding up the individual waiting times measures a different cost.  Motivated by the per-batch
maximum-delay objective introduced for multi-level aggregation by Lu, Ren,
Liu, and Tang~\cite{lu2026mlamax}, we replace the sum of all request delays by one
maximum waiting time for every
nonempty service batch.  If a walk $P$ executed at time $t$ serves the
nonempty request set $B$, its cost is
\begin{equation}
  \label{eq:intro-objective}
  \len(P)+\max_{q\in B}\bigl(t-a(q)\bigr).
\end{equation}
Here movement is instantaneous, and $a(q)$ is the arrival time of request $q$.
The total objective is the sum of~\eqref{eq:intro-objective} over service
events, together with the distance traveled without serving any request.
This change raises several questions.

Can the offline optimum be computed efficiently even though the endpoint of
each service walk determines where the next walk starts?  Does it matter whether the
server may encounter a request without serving it?  Can one still obtain a
polylogarithmic competitive ratio, or even a constant one, on a line, a tree,
or an arbitrary metric?  Finally, what remains possible if arrivals are
announced but their locations stay hidden until the server discovers them?

The changing server position distinguishes this problem from maximum-delay
aggregation at a fixed root. In rooted aggregation, the cost of serving a
set of requests is independent of earlier services. Here each service walk
ends at the starting point of the next walk. The cost of serving a set
therefore also depends on where the server was left. An analysis based only on the cost of each
request set therefore does not directly apply.

We address this dependence using \emph{certificates for groups of requests}.
The groups arrive in ordered, disjoint time windows. We assign each a nonnegative
value bounded by both the window length and the minimum weight of a tree
connecting its locations, possibly through other points of the metric.
Together, these values give a lower bound on the offline cost.

To prove the bound, we first rearrange the offline schedule into batches of
consecutive arrival epochs, where an epoch consists of all requests
arriving at the same time. Groups strictly between the first and
last groups represented in a batch are charged to that batch's delay.
For every remaining group, we use the part of the offline walk between its
first and last service. This part visits all locations in the group, and
each movement segment is counted at most twice.
\emph{Interior groups are charged to delay; the remaining groups are charged to movement.}
Our online algorithms alternate between waiting and service; we call each
such period a phase. Splitting the groups they construct according to odd
and even phases gives two families, each using disjoint requests.
If the groups and their windows satisfy the certificate theorem and the
online cost is at most $\alpha\sum_h C_h$, we obtain
\begin{equation}
  \label{eq:intro-certificate-transfer}
  \ALG
    \le \alpha\sum_h C_h
    \le 4\alpha\OPT.
\end{equation}
Every online upper bound uses this inequality; $\alpha$ bounds the cost of
service or search. Figure~\ref{fig:certificate-engine} illustrates the
charges to delay and movement.

\paragraph{Our contributions.}
We work with known finite metrics and instantaneous movement; a request's
waiting time is the time elapsed since its arrival.
Under \emph{elective service}, the server may visit a request without serving
it. Under \emph{automatic service}, every visit serves the pending requests
there. All online upper bounds are deterministic and hold under both rules.

\begin{enumerate}
  \item \emph{The certificate bound.}
  For the groups and values $C_k$ described above,
  Theorem~\ref{thm:certificate-engine} proves that every feasible elective
  schedule $\sigma$ satisfies
  \[
    \sum_kC_k
      \le \operatorname{Delay}(\sigma)+2\operatorname{Mov}(\sigma)
      \le2\cost(\sigma).
  \]
  The odd--even split then gives~\eqref{eq:intro-certificate-transfer}.

  \item \emph{Online algorithms on lines, trees, and general metrics.}
  We obtain ratios $10$ on lines, $12$ on weighted trees, and $20$ on
  arbitrary metrics in polynomial time (Table~\ref{tab:intro-engine}).
  For general metrics, the certificate is half the largest terminal MST
  weight seen during a phase; the running maximum is needed because MST
  weight may decrease when a request arrives. An exact Steiner tree oracle
  gives ratio $12$ on arbitrary metrics.

  \item \emph{Equality of the two offline optima.}
  An elective schedule can be rearranged into batches of consecutive
  arrival epochs without increasing its cost. Following the
  same walks under automatic service may split these batches but does not
  increase the cost. Thus the two offline optimal costs are equal
  (Theorem~\ref{thm:auto-elective-equivalence}). We also describe how each
  online algorithm runs under automatic service.

  \item \emph{Offline computation and hardness.}
  The common offline value is computable exactly in polynomial time on lines
  and explicitly represented weighted trees. Exact optimization on an
  arbitrary finite metric is NP hard even when all requests arrive at once.

  \item \emph{Service with hidden locations.}
  In the blind model, request labels and arrival times are announced, but
  locations remain hidden until visited. On a known finite line, searching
  with doubling radii gives $\alpha=21$ and competitive ratio $84$.
  On a unit star with $d$ leaves, even one hidden request gives lower bounds
  $2d-1$ for deterministic algorithms and $d$ for randomized algorithms
  against an oblivious adversary.

  \item \emph{Lower bounds.}
  Every deterministic visible algorithm has ratio at least $3$ already on a
  fixed line with two sites. For one blind request on a line, the
  deterministic and randomized lower bounds are $3$ and $2$, respectively.
\end{enumerate}

\begin{table}[htbp]
  \centering
  \small
  \setlength{\tabcolsep}{3pt}
  \begin{tabular}{@{}p{0.23\linewidth}p{0.31\linewidth}ccp{0.15\linewidth}@{}}
    \toprule
    Setting & Phase certificate $C_h$ & $\alpha$ & Ratio & Computation \\
    \midrule
    Visible finite line & length spanning two successive service intervals & $5/2$ & $10$ & polynomial \\
    Visible weighted tree & minimal subtree weight & $3$ & $12$ & polynomial \\
    Visible finite metric & minimum Steiner tree weight & $3$ & $12$ & oracle \\
    Visible finite metric & \raggedright half the largest MST weight seen in the phase & $5$ & $20$ & polynomial \\
    Blind finite line & \raggedright distance between successive phase endpoints & $21$ & $84$ & polynomial \\
    \bottomrule
  \end{tabular}
  \caption{Online bounds from \eqref{eq:intro-certificate-transfer}.
  All ratios hold for elective and automatic service; the exact Steiner
  tree row assumes an oracle.}
  \label{tab:intro-engine}
\end{table}

\paragraph{Maximum versus additive delay.}
Under additive delay, many young requests in one region can accumulate more
waiting cost than a few older requests elsewhere. Existing algorithms
coordinate these local costs with the distances needed to serve the
requests~\cite{azar2017osd,touitou2023improved}.
With maximum delay per batch, younger requests do not increase the delay
charge of a batch containing its oldest request. Our visible algorithms
use that oldest request's age as a common clock and serve the phase's
requests together, with a threshold determined by the metric.
The certificate bound shows that this choice loses only a constant factor
against an optimum that may serve different subsets at different times.

\paragraph{Related work.}
\emph{Online service and metric optimization with delay.}
Azar et al.~\cite{azar2017osd} introduced online service with additive delay
and gave a polylogarithmic randomized algorithm on general metrics, together
with multi-server extensions.  Azar and Touitou improved the general-metric
ratio to $O(\log^2 N)$ through a broader framework for metric optimization
with delay and deadlines~\cite{azar2019framework}.  Touitou subsequently
obtained a deterministic polynomial-time $O(\log N)$ algorithm
\cite{touitou2023improved}.  On a line, Bie{\'n}kowski et al.
gave an $O(\log N)$ Bucket algorithm~\cite{bienkowski2018osdline}.  A recent
logarithmic lower bound applies to a broad lazy-server class containing the
best known classical algorithms, while the existence of an unrestricted
constant-competitive algorithm remains unresolved
\cite{disser2025lazy}.
In the nonclairvoyant variant, future delay costs are hidden but locations
are known. Touitou matched the known square-root lower bounds up to
logarithmic factors~\cite{touitou2025nonclairvoyant}.
Our blind model reveals arrival times but hides locations until visited
(Section~\ref{sec:blind-line}).

\emph{Matching, caching, and multiple servers with delay.}
Emek et al. introduced min-cost perfect matching with
delays, whose objective combines metric connection cost with the sum of
request waiting times~\cite{emek2016matching}.  Subsequent work obtained an
$O(\log N)$ randomized bound on finite metrics, developed the bipartite
variant and nearly matching lower bounds, and extended the model to concave
delay functions~\cite{azar2017matching,ashlagi2017bipartite,
azar2021concave}.  More recently, Dufay and Wattenhofer obtained a
deterministic polylogarithmic guarantee without requiring the metric space to
be known in advance~\cite{dufay2026matching}.  In paging with delay, requests
may remain unserved while accumulating penalties; Gupta et al.
gave $O(\log k\log N)$-competitive algorithms and constant-factor
offline approximations, together with APX-hardness
\cite{gupta2022caching}. For online service with $k$ servers on the uniform
metric, Krneti\'c et al. proved a tight deterministic ratio of
$2k+1$~\cite{krnetic2020kserver}. These delay objectives add the delays of individual requests, whereas ours
charges the maximum waiting time in each service batch.

\emph{Per-batch and nonadditive delay.}
Dynamic TCP acknowledgment is the canonical fixed-location batching problem;
its deterministic and randomized guarantees are studied in
\cite{dooly2001tcp,karlin2001dynamic,seiden2000guessing}.  Albers and Bals
considered TCP objectives penalizing long delays
\cite{albers2005dynamic}, and Bhore et al.
developed a general theory of batch-aware and batch-oblivious TCP delay
functions~\cite{bhore2026online}. Lu et al. isolate the same per-batch
maximum objective for multi-level aggregation and obtain optimal deterministic
and randomized guarantees for static rooted services~\cite{lu2026mlamax}.
In our problem, each service walk determines where the next walk starts.
Our certificate bound accounts for this dependence by using the walks of
the complete offline schedule.

\paragraph{Organization.}
Section~\ref{sec:model} defines the model, and Section~\ref{sec:certificate-framework}
proves the certificate bound. Section~\ref{sec:online-overview} explains its online
applications. Sections~\ref{sec:visible-details} and~\ref{sec:blind-line} give the
remaining analysis. Section~\ref{sec:offline-results} treats offline algorithms
and lower bounds, and Section~\ref{sec:conclusion} discusses open questions.

\FloatBarrier
\section{Preliminaries}
\label{sec:model}

\paragraph{Requests and costs.}

Let $(X,d)$ denote a known finite metric space with $N=|X|$, and let $s_0\in X$ denote the initial server position.
Each request $q$ has an arrival time $a(q)\ge0$ and a location $x(q)\in X$.
The input contains finitely many requests, which are distinct even if they have the same arrival time and location.
An \emph{arrival epoch} consists of all requests arriving at one time.
Write the distinct arrival times as
\begin{equation}
\label{eq:model-display-1}a_1<a_2<\cdots<a_m,
\end{equation}
and let $R_h$ denote the requests arriving at time $a_h$.
All arrivals at a timestamp are processed before any service at that time.

Movement is instantaneous but is charged by distance.
At time $t$, the server may follow a finite walk $P=(v_0,v_1,\ldots,v_\ell)$ in $X$ and serve a nonempty set $B$ of pending requests at its listed vertices.
The walk starts at the current server position and has length
\begin{equation}
\label{eq:model-display-2}\len(P)=\sum_{i=1}^{\ell}d(v_{i-1},v_i).
\end{equation}
The requests in $B$ form one \emph{batch}.
Their waiting times increase at rate one. The cost of this service event is
\begin{equation}
\label{eq:event-cost}\cost(P,B,t)=\len(P)+\max_{q\in B}\bigl(t-a(q)\bigr).
\end{equation}
Movement that serves no request costs only its length.

A schedule is a finite sequence of walks in time order, starting at $s_0$; each walk starts where the preceding one ended.
The server stays at that point between walks.
Several walks may occur at the same timestamp, each charged as a separate event, with no new arrival processed between them.
Every request must eventually be served.
For a schedule $\sigma$, write
\begin{equation}
\label{eq:model-display-3}\operatorname{Mov}(\sigma)=\sum_P\len(P),\qquad\operatorname{Delay}(\sigma)=\sum_{(P,B,t):B\ne\varnothing}\max_{q\in B}(t-a(q)).
\end{equation}
The total cost is $\cost(\sigma)=\operatorname{Mov}(\sigma)+\operatorname{Delay}(\sigma)$.

\paragraph{Elective and automatic service.}

We consider two rules for serving requests at visited locations.
\begin{definition}[Elective service]
 \label{def:elective-service}
When the server visits or occupies a pending request's location, it may choose whether to serve that request.
\end{definition}

\begin{definition}[Automatic service]
 \label{def:automatic-service}
Every pending request at a vertex visited by $P$ is served in the current event.
Requests arriving at the current server position are served immediately in one batch, after processing the arrival epoch and before any further movement at that time.
A walk intended only to change position must also serve any pending requests it visits.
\end{definition}

Serving an old request early may create a separate delay charge, meaning that the two rules allow different decisions.
Nevertheless, their offline optimal costs are equal, as we prove in Section~\ref{sec:automatic-service}.

\paragraph{Competitive guarantees.}

Unless stated otherwise, requests are \emph{visible}: their locations are revealed at arrival.
The offline optimum knows all arrival times and locations.
Let $\OPT(I)$ denote the infimum cost of a feasible schedule for input $I$ under the specified service rule, and let $\ALG(I)$ denote the online algorithm's cost.
We omit $I$ when it is clear.
A deterministic algorithm has competitive ratio at most $\rho$ if $\ALG(I)\le \rho\OPT(I)$ for every finite input $I$.
No additive constant is included unless explicitly stated.

\FloatBarrier
\section{Offline Structure and Certificate Bounds}
\label{sec:certificate-framework}

We first show that an elective offline schedule can be rearranged into batches of consecutive arrival epochs without increasing movement or delay.
Following the same walks under automatic service proves that the two offline optima are equal.
We then assign certificates to groups of requests and bound their sum using the offline schedule's waiting and movement costs.
The analysis of the online algorithms in the following sections uses this bound.

\subsection{Normalization of elective schedules}
\label{sec:offline-normal-form}

We introduce a normalization of elective schedules in which each batch contains all requests from a consecutive sequence of arrival epochs.
The change repeatedly combines two adjacent batches and works on any finite metric.

For a nonempty request set $B$, let $\alpha(B)$ and $\beta(B)$ denote its earliest and latest arrival times, respectively, with
\begin{equation}
\label{eq:offline-normal-form-display-1}\alpha(B)=\min_{q\in B}a(q),\qquad\beta(B)=\max_{q\in B}a(q).
\end{equation}
A schedule with the properties in the next proposition is called \emph{normalized}.
\begin{proposition}
 \label{prop:elective-normal-form}
For every feasible elective schedule, there exists a schedule with no greater movement or delay whose nonempty service batches are consecutive blocks of complete arrival epochs.
A block beginning at epoch $p$ and ending at epoch $j$ is served at time $a_j$.
\end{proposition}

\begin{proof}
Consider consecutive batches $B$ and $C$ served at times $t\le u$, respectively.
If $\alpha(C)\le t$, replace them by one batch as follows.
Wait at the start of $B$'s walk until $u$, then follow both walks and the movement between them, serving $B\cup C$.
Every request in these batches has arrived by $u$.
The movement length and final position are unchanged.
The old and merged delays are
\[
(t-\alpha(B))+(u-\alpha(C))\qquad\text{and}\qquad u-\min\{\alpha(B),\alpha(C)\},
\]
respectively.
Their difference is $t-\max\{\alpha(B),\alpha(C)\}\ge0$, implying that merging adds no cost.

Repeat this operation while an eligible pair remains.
Each operation removes one batch. The process therefore terminates with
\begin{equation}
\label{eq:offline-normal-form-display-4}\beta(B_k)\le t_k<\alpha(B_{k+1}),
\end{equation}
for every consecutive pair.
This means that every request in $B_k$ arrives before every request in $B_{k+1}$.
Since the batches contain all requests, each batch consists of consecutive complete arrival epochs.
Now execute the walk of each batch, including all movement since the preceding batch, at its latest arrival time, i.e., set $t_k=\beta(B_k)$ for each $B_k$; this can only decrease each batch's delay cost.
Delete any movement after the last batch, and the proposition follows.
\end{proof}

\subsection{Equivalence of offline service values}
\label{sec:automatic-service}

We next compare the two service rules.
Under elective service, requests arriving at the parked server may wait to be served in a later batch; under automatic service, they are served immediately.
These early services can split a normalized batch, and an optimal automatic schedule may have to serve requests in batches that are not consecutive arrival blocks.\footnote{For instance, start at $0$ on the two-point line $X=\{0,d\}$, where $d>3$, with requests $A=(d,0)$, $B=(0,1)$, $C=(0,2)$, and $D=(d,3)$, listed as (location, arrival time).
Waiting until time $3$ and then moving to $d$ costs $d+3$ under either rule.
Elective service permits one batch $\{A,B,C,D\}$, whereas automatic service gives batches $\{B\}$, $\{C\}$, and $\{A,D\}$, with $B,C$ served at age zero.
For automatic batches to be consecutive arrival blocks, $A$ must be served before $B$ arrives at time $1$; otherwise $B$ is served immediately at $0$ while $A$ is still pending.
Serving $A$ and then $B$ requires movement $0\to d\to0$, costing at least $2d>d+3$.}
Nevertheless, the walks of a normalized elective schedule can also be used under automatic service without increasing cost, hence giving the same offline optimal value.
\begin{theorem}
 \label{thm:auto-elective-equivalence}
For every finite input on a finite metric under the instantaneous movement, event ordering, and common waiting rate of Section~\ref{sec:model}, $\OPT_{\mathrm{auto}}=\OPT_{\mathrm{elective}}$.
\end{theorem}

\begin{proof}
Every automatic schedule is elective, implying $\OPT_{\mathrm{elective}}\le\OPT_{\mathrm{auto}}$.
To prove $\OPT_{\mathrm{auto}}\le\OPT_{\mathrm{elective}}$, we show that any elective schedule can be converted to an automatic schedule without increasing cost.
In particular, any given elective schedule is first normalized using Proposition~\ref{prop:elective-normal-form}.
In the normalized schedule, each block $B_k$ is served by a walk $P_k$ at time $t_k=\beta(B_k)$, with $\beta(B_k)<\alpha(B_{k+1})$.
A corresponding automatic schedule follows the same walks at the same times.
Before the planned walk for a batch, its requests can only have been served at arrival at the parked server.
For $j<k$, we have $t_j=\beta(B_j)<\alpha(B_{j+1})\le\alpha(B_k)$, implying that no earlier walk can visit a request of $B_k$ after it has arrived.
Also, no request of a later block has arrived at $t_k$.
The walk $P_k$ therefore serves exactly those members of $B_k$ that were not served for free at the parked server.
Its delay charge is at most $t_k-\alpha(B_k)$, and its movement length is unchanged.
If all members of a block have already been served, the automatic server still follows its walk to reach the same endpoint.
See Algorithm~\ref{alg:automatic-replay} in Appendix~\ref{app:automatic-replay} for more details of the conversion.
The resulting automatic schedule costs no more than the elective schedule, implying $\OPT_{\mathrm{auto}}\le\OPT_{\mathrm{elective}}$.
\end{proof}

\subsection{Certificate bounds for ordered request groups}
\label{sec:metric-certificates}

We now prove the lower bound used to analyze the online algorithms.
Each group has a value bounded by both the length of its arrival window and the cost of connecting its locations.
The two bounds allow us to charge some groups to offline waiting and the remaining groups to offline movement.

For a set of locations $A\subseteq X$, we need the minimum weight of a tree connecting them.
The tree may also use other points of $X$.
Its \emph{Steiner cost} is
\begin{equation}
\label{eq:ambient-steiner-cost}\ST_X(A)=\min_F\sum_{e\in E(F)}d(e),
\end{equation}
where $F$ ranges over trees in the complete graph on $X$ with $A\subseteq V(F)$, and $d(e)=d(u,v)$ for $e=\{u,v\}$.
Vertices of $X\setminus A$ may be used as Steiner vertices.
A walk visiting all points of $A$ contains a connected graph spanning $A$.
Removing repeated edges and cycles gives a tree of no greater weight.
Thus every such walk has length at least $\ST_X(A)$.

The following theorem bounds the sum of certificates for one sequence of groups.
\begin{theorem}[Certificate bound for request groups]
\label{thm:certificate-engine}
Let $0=T_0<T_1<\cdots<T_m$ denote window endpoints, with $I_1=[T_0,T_1]$ and $I_k=(T_{k-1},T_k]$ for $k\ge2$.
The first window may contain one free dummy request at the initial server position and time $T_0$.
Consider an instance partitioned into nonempty request groups $G_1,\ldots,G_m$, with every request of $G_k$ arriving in $I_k$.
For each $k$, let $C_k$ denote a nonnegative value satisfying
 \begin{equation}
\label{eq:metric-group-certificate-assumptions}0\le C_k\le T_k-T_{k-1}\qquad\text{and}\qquad C_k\le\ST_X\bigl(\{x(q):q\in G_k\}\bigr).
\end{equation}
Then every feasible elective schedule $\sigma$ satisfies
 \begin{equation}
\label{eq:metric-group-matching-bound}\sum_{k=1}^m C_k\le\operatorname{Delay}(\sigma)+2\operatorname{Mov}(\sigma)\le 2\cost(\sigma).
\end{equation}
Consequently, $\sum_kC_k\le2\OPT$.
\end{theorem}

\begin{proof}
Given any feasible schedule $\sigma_0$, which may use automatic service, we regard it as an elective schedule and obtain a normalized elective schedule $\sigma$ using Proposition~\ref{prop:elective-normal-form}.
Neither movement nor delay increases.
We then charge interior groups to delay and all remaining groups to movement, counting each movement segment at most twice.
Adding these charges gives the claimed bound.
The request groups remain as given in the theorem.
Movement between batches is included in the walk of the next batch, as in that proposition.

For a batch $B$ in $\sigma$, we compare its delay charge with the certificates of the groups it contains.
Let $p(B)$ and $q(B)$ denote the smallest and largest indices of groups with requests in $B$, and let $d(B)$ denote its delay charge.
A group $G_k$ with $p(B)<k<q(B)$ is called an \emph{interior group} of $B$.
The groups are ordered by arrival time, and $B$ contains consecutive complete arrival epochs.
Thus every request in an interior group belongs to $B$.
If there is an interior group, the batch contains an arrival no later than $T_{p(B)}$ and an arrival after $T_{q(B)-1}$.
Consequently,
\begin{align}
d(B)\ge T_{q(B)-1}-T_{p(B)}=\sum_{k=p(B)+1}^{q(B)-1}(T_k-T_{k-1})\ge\sum_{k=p(B)+1}^{q(B)-1}C_k.\label{eq:metric-group-interior-delay-charge}
\end{align}
Each interior group belongs to only one batch, implying that these charges total at most $\operatorname{Delay}(\sigma)$.

For the remaining groups, we use movement to bound their certificates.
For each such group $G_k$, let $Q_k$ denote the part of the complete offline walk in $\sigma$ from the first service of a request in $G_k$ to the last.
Different traversals of the same edge are separate occurrences in this walk.
Since $Q_k$ visits all locations of $G_k$,
\begin{equation}
\label{eq:metric-group-connector-lower-bound}\len(Q_k)\ge\ST_X\bigl(\{x(q):q\in G_k\}\bigr)\ge C_k.
\end{equation}
For the batches $B_1,\ldots,B_r$ of $\sigma$, let $J_k$ denote the set of batch indices that meet $G_k$, namely,
\begin{equation}
\label{eq:metric-certificates-display-3}J_k=\{i\in[r]:B_i\cap G_k\ne\varnothing\}.
\end{equation}
The set $J_k$ is an interval of batch indices.
Indeed, if $G_k$ has requests in two batches, then every batch between them must also contain requests of $G_k$, by arrival order.
For a movement segment in the walk serving $B_i$, we count the walks $Q_k$ containing it.
If $Q_k$ contains this segment, its first service is in $B_i$ or an earlier batch, and its last service is in $B_i$ or a later batch.
The interval property gives $i\in J_k$.
Since $G_k$ was not an interior group of $B_i$, it must satisfy
\begin{equation}
\label{eq:metric-certificates-display-4}k=p(B_i)\qquad\text{or}\qquad k=q(B_i),
\end{equation}
implying that there are at most two such group indices.
Thus each movement segment is counted in at most two of the remaining walks $Q_k$.
All movement is included in batch walks, implying
\begin{equation}
\label{eq:metric-certificates-display-9}\sum_{k\text{ not charged to delay}} C_k\le \sum_{k\text{ not charged to delay}}\len(Q_k)\le 2\operatorname{Mov}(\sigma).
\end{equation}

Adding the charges gives
\begin{equation}
\label{eq:metric-certificates-display-5}\sum_k C_k\le \operatorname{Delay}(\sigma)+2\operatorname{Mov}(\sigma)\le \operatorname{Delay}(\sigma_0)+2\operatorname{Mov}(\sigma_0),
\end{equation}
and
\begin{equation}
\label{eq:metric-certificates-display-6}\operatorname{Delay}(\sigma_0)+2\operatorname{Mov}(\sigma_0)\le 2\bigl(\operatorname{Delay}(\sigma_0)+\operatorname{Mov}(\sigma_0)\bigr).
\end{equation}
Taking the infimum gives $\sum_k C_k\le2\OPT$.
The optional dummy is served for free at time zero.
\end{proof}

\begin{corollary}[Competitive bound from two families]
\label{cor:certificate-transfer}
Suppose an online execution produces request groups $G_h$, values $C_h\ge0$, and a partition of their indices into two colors.
For each color, assume that the groups have disjoint request identities and admit windows satisfying Theorem~\ref{thm:certificate-engine} on the subinstance containing exactly those groups.
The optional initial dummy is allowed.
If, for some $\alpha\ge0$,
 \begin{equation}
\label{eq:local-realization-bound}\ALG\le\alpha\sum_h C_h,
\end{equation}
then
 \begin{equation}
\label{eq:four-alpha-transfer}\ALG\le4\alpha\OPT.
\end{equation}
\end{corollary}

\begin{proof}
For either color, the subinstance contains only the requests used by that color.
Removing the other requests from an elective schedule cannot increase its cost, implying that the optimum for this smaller input is at most $\OPT$.
Theorem~\ref{thm:certificate-engine} bounds that color's certificate sum by $2\OPT$.
Adding the bounds for the two colors gives
\begin{equation}
\label{eq:metric-certificates-display-7}\sum_hC_h\le4\OPT.
\end{equation}
Combining this with~\eqref{eq:local-realization-bound} finishes the proof.
\end{proof}

Figure~\ref{fig:certificate-engine} illustrates the two color classes and the delay and movement charges for one class.

\begin{figure}[ht]
  \centering
  \begin{tikzpicture}[
      x=0.90cm,
      y=0.85cm,
      >=Latex,
      every node/.style={font=\scriptsize},
      axis/.style={draw=black!65,thin,-{Latex[length=1.8mm]}},
      odd/.style={draw=blue!65!black,fill=blue!12,thick},
      even/.style={draw=orange!85!black,fill=orange!14,thick},
      temporal/.style={draw=green!45!black,fill=green!12,thick},
      spatial/.style={draw=red!70!black,fill=red!10,thick}
    ]

    \begin{scope}
      \node[font=\small\bfseries] at (3,4.6) {(a) Odd and even phases};
      \draw[axis] (0,1.5)--(6.35,1.5);
      \foreach \k in {0,...,6} {
        \draw (\k,1.57)--(\k,1.43);
        \node[below=3pt] at (\k,1.43) {$\tau_{\k}$};
      }
      \node at (3,0.65) {phase time};
      \node[anchor=east,text=blue!65!black] at (-0.15,3.55) {odd};
      \node[anchor=east,text=orange!85!black] at (-0.15,2.55) {even};
      \draw[odd] (0,3.25) rectangle (1,3.85);
      \node at (0.5,3.55) {$C_1$};
      \draw[odd] (1,3.25) rectangle (3,3.85);
      \node at (2,3.55) {$C_3$};
      \draw[odd] (3,3.25) rectangle (5,3.85);
      \node at (4,3.55) {$C_5$};
      \draw[even] (0,2.25) rectangle (2,2.85);
      \node at (1,2.55) {$C_2$};
      \draw[even] (2,2.25) rectangle (4,2.85);
      \node at (3,2.55) {$C_4$};
      \draw[even] (4,2.25) rectangle (6,2.85);
      \node at (5,2.55) {$C_6$};
      \node[align=center,text=black!70] at (3,-0.15)
        {windows are disjoint within each color;\\each color is compared separately with $\OPT$};
    \end{scope}

    \begin{scope}[xshift=8.0cm]
      \node[font=\small\bfseries] at (3.15,4.6) {(b) Charging offline batch $B$};
      \draw[very thick,black!60] (0,3.95)--(6.3,3.95);

      \draw[spatial] (0,2.9) rectangle (1,3.6);
      \node at (0.5,3.25) {$G_p$};
      \draw[temporal] (1.35,2.9) rectangle (4.95,3.6);
      \node[align=center] at (3.15,3.25)
        {$G_{p+1},\ldots,G_{q-1}$};
      \draw[spatial] (5.3,2.9) rectangle (6.3,3.6);
      \node at (5.8,3.25) {$G_q$};

      \draw[green!45!black,thick,-{Latex[length=1.8mm]}]
        (3.15,2.87)--(3.15,2.15);
      \node[align=center,text=green!35!black] at (3.15,1.7)
        {interior groups\\charged to delay $d(B)$};

      \draw[red!70!black,thick,-{Latex[length=1.8mm]}]
        (0.5,2.87)--(0.5,2.15);
      \draw[red!70!black,thick,-{Latex[length=1.8mm]}]
        (5.8,2.87)--(5.8,2.15);
      \draw[red!70!black,thick] (0,1.95)--(1,1.95)
        node[midway,below=4pt] {$Q_p$};
      \draw[red!70!black,thick] (5.3,1.95)--(6.3,1.95)
        node[midway,below=4pt] {$Q_q$};
      \node[align=center,text=red!60!black] at (3.15,0.3)
        {remaining groups charged to movement;\\each segment counted at most twice};
    \end{scope}
  \end{tikzpicture}
  \caption{The certificate bound. In (a), groups of the same color use
  disjoint requests and disjoint arrival windows. In (b), $p$ and $q$ are
  the first and last group indices in an offline batch $B$.
  Its delay pays for the groups strictly between them.
  Among groups not charged to delay, only $G_p$ and $G_q$ can use movement
  in this batch's walk. Thus each movement segment is counted at most
  twice in Theorem~\ref{thm:certificate-engine}.}
  \label{fig:certificate-engine}
\end{figure}

In each application below, we specify the groups and their windows, check the two certificate inequalities in \eqref{eq:metric-group-certificate-assumptions}, and bound the algorithm's cost as in \eqref{eq:local-realization-bound}.
The next remark allows us to use only the requests needed for these groups.
\begin{remark}
 \label{rem:metric-certificate-subinstances}
Suppose $I'$ contains only some of the requests in an input $I$.
An elective schedule for $I$ can serve $I'$ by following the same walks and omitting the other requests.
Its cost cannot increase.
Therefore, if the groups in $I'$ satisfy Theorem~\ref{thm:certificate-engine}, then
 \begin{equation}
\label{eq:metric-certificates-display-8}\sum_k C_k\le2\OPT(I')\le2\OPT(I).
\end{equation}
A request may occur in both color classes, because each class is compared separately with the full optimum.
Within a color class, the groups must remain disjoint.
\end{remark}

\clearpage
\section{Online Algorithms from the Certificate Bound}
\label{sec:online-overview}

The certificate bound separates two tasks: choosing when and how to serve requests, and comparing the resulting cost with the offline optimum.
This section explains the choices for trees, general metrics, and lines, including hidden request locations.
For each algorithm, we specify its service rule, its request groups, and the inequalities needed by Corollary~\ref{cor:certificate-transfer}.
The subsequent sections provide the remaining implementation details and proofs.

\subsection{Phase notation}
\label{subsec:overview-phases}
\label{subsec:execution-conventions}

Each algorithm operates in phases and clears all requests assigned to a phase before starting the next one.\footnote{A batch consists of the requests served in one event. A phase includes waiting and may contain several walks. Apart from free service at the current position, visible phases serve one batch, whereas blind automatic phases may serve several.}
For phase $h$, let $P_h$, $\tau_h$, and $s_h$ denote its request set, completion time, and final server position, respectively.
We omit zero-cost phases; the remaining completion times are strictly increasing, with initial time $\tau_0=0$.
Requests assigned to phase $h$ arrive after $\tau_{h-1}$ and by $\tau_h$, except that the first phase may include arrivals at time zero.
Each phase ends at the location of a request $p_h\in P_h$.
A free dummy request $p_0$ at $s_0$ represents the initial position, with $P_0=\{p_0\}$.
All arrivals are processed before a tied service event, and new arrivals interrupt the algorithm's waiting timers.

\subsection{Visible service on trees and general metrics}
\label{subsec:overview-connecting-trees}

\paragraph{Weighted trees: the algorithm.}
Let $T=(V,E)$ denote a finite tree with positive edge weights; requests and the server lie at vertices.
In phase $h$, the server waits at its current vertex $s=s_{h-1}$, called the \emph{anchor}.
Requests arriving at $s$ are served immediately at zero cost; all other requests join the pending set.
Let $\alpha_h$ denote the first such request's arrival time.
At time $t$, let $A(t)$ denote $s$ together with the pending locations, and let $W(t)$ denote the weight of the smallest subtree $T[A(t)]$ connecting them.
This subtree is the union of the paths from $s$ to those locations, with each edge counted once.
The server waits until the oldest pending request's age reaches this weight:
\[
t-\alpha_h=W(t).
\]
The tentative service time is therefore $\alpha_h+W(t)$.
Each new arrival updates the subtree without resetting $\alpha_h$.
An arrival within the existing subtree leaves the service time unchanged; one that adds edges postpones it.
Since $W(t)$ never decreases, the first crossing occurs at equality.
After the final arrival, $W(t)$ is fixed and the server eventually serves.

At time $\tau_h$, one instantaneous walk visits the entire subtree and serves all pending requests as batch $P_h$.
It chooses a farthest pending request $p_h$ and stops at its location $s_h$.
An open depth first traversal does this by exploring the other branches before finishing along the path to $s_h$.
Every edge on the $s$--$s_h$ path is used once; every other edge is used twice.
Writing $W_h=W(\tau_h)$, the movement is at most $2W_h$ and the batch delay is $\tau_h-\alpha_h=W_h$.
Thus $\ALG_h\le3W_h$.
The next phase starts at $s_h$ with no pending requests.
Figure~\ref{fig:tree-phase-example} illustrates one phase.

\begin{figure}[ht]
  \centering
  \begin{minipage}[c]{0.36\textwidth}
    \centering
    \begin{tikzpicture}[x=1cm,y=1cm,every node/.style={font=\small}]
      \node[circle,draw,fill=orange!12,inner sep=2pt] (s) at (0,1.25) {$s$};
      \node[circle,draw,fill=blue!8,inner sep=2pt] (a) at (-1.1,0) {$a$};
      \node[circle,draw,fill=blue!8,inner sep=2pt] (b) at (1.1,0) {$b$};
      \draw[thick] (s) -- node[left=4pt] {$2$} (a);
      \draw[thick] (s) -- node[right=4pt] {$3$} (b);
      \node[below=6pt,font=\footnotesize] at (a.south) {arrivals: $0,3$};
      \node[below=6pt,font=\footnotesize] at (b.south) {arrival: $1$};
    \end{tikzpicture}
  \end{minipage}\hfill
  \begin{minipage}[c]{0.61\textwidth}
    \centering\small
    \begin{tabular}{@{}clcc@{}}
      \toprule
      Time & Arrival & $W(t)$ & Service time\\
      \midrule
      $0$ & at $a$ & $2$ & $2$ (tentative)\\
      $1$ & at $b$ & $5$ & $5$ (tentative)\\
      $3$ & at $a$ & $5$ & $5$ (unchanged)\\
      $5$ & none & $5$ & service occurs\\
      \bottomrule
    \end{tabular}
  \end{minipage}
  \caption{One tree phase starting at $s$. The first request arrives at time zero. The arrival at $b$ increases the subtree weight and postpones service; the later arrival at $a$ does neither. At time $5$, the walk $s\to a\to s\to b$ serves all three requests. Movement costs $7$ and maximum delay costs $5$, giving total cost $12\le3W_h=15$.}
  \label{fig:tree-phase-example}
\end{figure}

\paragraph{Weighted trees: comparison with the optimum.}
The pending requests alone may have zero Steiner cost: if they all lie at one point $x\ne s$, reaching them still costs $d(s,x)$.
We therefore include the preceding endpoint request $p_{h-1}$, whose location is $s=s_{h-1}$.
The group $G_h=\{p_{h-1}\}\cup P_h$ spans exactly the subtree traversed in phase $h$, giving certificate $C_h=W_h$.
The request $p_{h-1}$ is used only in the analysis; it is not served again.

The arrival window must include both the current requests and $p_{h-1}$, which arrived in the preceding phase.
We take $I_h=(\tau_{h-2},\tau_h]$ for $h\ge3$, with $I_1=[0,\tau_1]$ and $I_2=[0,\tau_2]$ initially.
Every request in $G_h$ arrives in $I_h$.
Moreover, $\alpha_h\ge\tau_{h-1}$ and $\tau_h-\alpha_h=W_h$, giving $|I_h|\ge\tau_h-\tau_{h-1}\ge W_h$.
The certificate and cost conditions are therefore
\begin{equation}
\label{eq:overview-checks}
C_h=W_h=\ST_V(\{x(q):q\in G_h\}),\qquad C_h\le|I_h|,\qquad \ALG_h\le3C_h.
\end{equation}

Theorem~\ref{thm:certificate-engine} requires disjoint arrival windows.
Adjacent windows $I_h$ and $I_{h+1}$ overlap, which is why we separate odd and even indices.
Within either family the windows are disjoint.
Figure~\ref{fig:certificate-engine}(a) illustrates this division.
The groups are also disjoint in request identities: $G_h\subseteq P_{h-1}\cup P_h$, whereas $G_{h+2}\subseteq P_{h+1}\cup P_{h+2}$.
For each family, Theorem~\ref{thm:certificate-engine} bounds the certificate sum by twice the optimum on its requests, which is at most the full-input optimum by Remark~\ref{rem:metric-certificate-subinstances}.
Hence $\sum_{h\text{ odd}}W_h\le2\OPT$ and $\sum_{h\text{ even}}W_h\le2\OPT$.
Adding them and using the phase cost bound yields
\begin{equation}
\label{eq:overview-tree-ratio}
\ALG\le3\sum_hW_h\le3(2\OPT+2\OPT)=12\OPT.
\end{equation}
The ratio $3\cdot2\cdot2=12$ comes from the phase cost factor $3$, the two families, and the factor $2$ in the single-family certificate bound.
This is the two-family comparison of Corollary~\ref{cor:certificate-transfer}; the offline schedule need not use the same batches or service times.

\paragraph{General metrics.}
An exact Steiner tree on $A(t)$ gives the same argument and ratio $12$, but requires an exact oracle.
For a polynomial time algorithm, we instead use a minimum spanning tree on $A(t)$, called a terminal MST.
For every terminal set $A$, its weight satisfies $\MST(A)\le2\ST_X(A)$: doubling a Steiner tree and shortcutting its traversal to $A$ gives a spanning path of at most twice the weight.
There is one difficulty: the terminal MST weight can decrease when a request arrives.
On a unit star with three leaves, connecting the leaves alone costs $4$; adding the center as a terminal reduces the MST weight to $3$.
The current MST weight therefore does not directly give the nondecreasing threshold required above.
Figure~\ref{fig:mst-envelope} illustrates this issue.

\begin{figure}[ht]
  \centering
  \begin{tikzpicture}[
      x=0.78cm,
      y=0.78cm,
      >=Latex,
      every node/.style={font=\scriptsize},
      terminal/.style={circle,draw=blue!65!black,fill=blue!12,thick,minimum size=5mm,inner sep=0pt},
      center/.style={circle,draw=orange!85!black,fill=orange!15,thick,minimum size=5mm,inner sep=0pt},
      mst/.style={draw=blue!65!black,very thick},
      steiner/.style={draw=orange!85!black,densely dashed,thick},
      axis/.style={draw=black!65,thin,-{Latex[length=1.8mm]}}
    ]

    \begin{scope}
      \node[font=\small\bfseries] at (2.2,4.25) {(a) Before arrival};
      \coordinate (a) at (0.4,0.9);
      \coordinate (b) at (4,0.9);
      \coordinate (c) at (2.2,3.6);
      \coordinate (o) at (2.2,1.8);
      \draw[mst] (a)--node[below] {$2$}(b)--node[right] {$2$}(c);
      \draw[black!30,dotted] (c)--(a);
      \draw[steiner] (o)--(a) (o)--(b) (o)--(c);
      \node[terminal] at (a) {$a$};
      \node[terminal] at (b) {$b$};
      \node[terminal] at (c) {$c$};
      \node[circle,draw=black!45,fill=white,minimum size=4mm,inner sep=0pt] at (o) {$o$};
      \node[align=center] at (2.2,0.05)
        {$\MST(\{a,b,c\})=4$\\$\ST_X(\{a,b,c\})=3$};
    \end{scope}

    \begin{scope}[xshift=5.5cm]
      \node[font=\small\bfseries] at (2.2,4.25) {(b) Center terminal};
      \coordinate (a2) at (0.4,0.9);
      \coordinate (b2) at (4,0.9);
      \coordinate (c2) at (2.2,3.6);
      \coordinate (o2) at (2.2,1.8);
      \draw[mst] (o2)--node[below left] {$1$}(a2)
        (o2)--node[below right] {$1$}(b2)
        (o2)--node[right] {$1$}(c2);
      \node[terminal] at (a2) {$a$};
      \node[terminal] at (b2) {$b$};
      \node[terminal] at (c2) {$c$};
      \node[center] at (o2) {$o$};
      \node[align=center] at (2.2,0.05)
        {$\MST(\{a,b,c,o\})=3$\\the MST weight decreases};
    \end{scope}

    \begin{scope}[xshift=11cm]
      \node[font=\small\bfseries] at (2.2,4.25) {(c) Running maximum};
      \draw[axis] (0.2,0.55)--(4.35,0.55) node[right] {$t$};
      \draw[axis] (0.45,0.3)--(0.45,3.85);
      \draw[black!25,dotted] (0.45,1.8)--(4.1,1.8);
      \draw[black!25,dotted] (0.45,2.45)--(4.1,2.45);
      \draw[black!25,dotted] (0.45,3.25)--(4.1,3.25);
      \node[left] at (0.4,1.8) {$2$};
      \node[left] at (0.4,2.45) {$3$};
      \node[left] at (0.4,3.25) {$4$};
      \draw[mst] (0.55,3.25)--(2.2,3.25)--(2.2,2.45)--(4.05,2.45);
      \draw[orange!85!black,very thick] (0.55,3.25)--(4.05,3.25);
      \draw[green!45!black,very thick] (0.55,1.8)--(4.05,1.8);
      \draw[black!55,dashed] (2.2,0.55)--(2.2,3.55);
      \node[below=2pt,align=center] at (2.2,0.55) {$o$ arrives};
      \node[anchor=west,text=blue!65!black] at (2.45,2.15) {$M(t)$};
      \node[anchor=west,text=orange!85!black] at (2.45,3.47) {$E(t)=\max_{u\le t}M(u)$};
      \node[anchor=west,text=green!35!black] at (2.45,1.58) {$b(t)=E(t)/2$};
    \end{scope}
  \end{tikzpicture}
  \caption{Why we keep the largest MST weight seen in a phase.
  In the three-leaf unit-star metric, the terminal MST has weight
  $4$ before the center $o$ is itself a terminal, although the Steiner tree (dashed) has weight $3$.  Adding $o$ makes the terminal MST
  drop to $3$. The running maximum $E(t)$ stays at $4$.
  The threshold $b(t)=E(t)/2$ does not decrease and is at most the
  Steiner cost of the final set of terminals.}
  \label{fig:mst-envelope}
\end{figure}

The algorithm retains the largest MST weight observed during the current phase and uses half of it as the threshold:
\begin{equation}
\label{eq:overview-mst-threshold}
E_h(t)=\max_{u\le t}\MST(A(u)),\qquad b(t)=E_h(t)/2,
\end{equation}
where the maximum ranges over prefixes of this phase only.
At service time, the server traverses a current MST and ends at an actual phase request.
For the final terminal set $A_h$, every earlier set $A(u)$ is a subset of $A_h$.
Consequently, $\MST(A(u))\le2\ST_X(A(u))\le2\ST_X(A_h)$, implying $b_h\le\ST_X(A_h)$.
The same groups $G_h=\{p_{h-1}\}\cup P_h$ therefore give valid certificates $C_h=b_h$.
The current MST weighs at most $E_h=2b_h$, giving movement at most $4b_h$ and delay $b_h$:
\begin{equation}
\label{eq:overview-metric-ratio}
\ALG_h\le5b_h,\qquad \ALG\le5\sum_hb_h\le20\OPT.
\end{equation}
One MST computation per arrival and one traversal per phase give polynomial running time.

\paragraph{The two service rules.}
Apart from free service at the anchor, automatic service occurs only during the final traversal, which clears all remaining requests.
Thus movement and phase delay obey the same bounds under both rules.
Theorem~\ref{thm:auto-elective-equivalence} gives the same comparison with the automatic optimum.
Section~\ref{sec:visible-details} also treats arbitrary approximations to the connecting tree.

\subsection{A better bound on a line}
\label{subsec:overview-line}

The order of locations on a line permits ratio $10$.
Let $H_{h-1}$ denote the interval spanned by the preceding phase's requests, and let $H(t)$ denote the interval spanned by the current pending requests.
Initially $H_0=\{s_0\}$.
The server ends each phase at an endpoint of its interval.
Its waiting threshold is $b(t)=\operatorname{diam}(H(t))+\dist(H_{h-1},H(t))$.
When the threshold is reached, the server goes to the nearer endpoint of $H(t)$, visits the entire interval, and ends at the other endpoint.
The threshold is nondecreasing as $H(t)$ expands.

Let $H_h$ denote the current interval at service, and let $X_h=\operatorname{diam}(H_{h-1}\cup H_h)$.
The group $G_h$ consists of requests at the two extremes of this union, with the initial dummy available for phase one.
Its Steiner cost is $X_h$.
The threshold satisfies $b_h\le X_h$, and the service walk has length at most $3X_h/2$.
If the server starts inside $H_h$, the nearer endpoint is at distance at most $\operatorname{diam}(H_h)/2$; otherwise, the walk moves in one direction and costs at most $X_h$.
Unlike the tree construction, one phase alone need not last $X_h$.
The two consecutive phases do: their thresholds satisfy $b_{h-1}+b_h\ge X_h$ for $h\ge2$, while $b_1=X_1$.
Hence $|I_h|\ge X_h$, and Corollary~\ref{cor:certificate-transfer} gives $\ALG\le(5/2)\sum_hX_h\le10\OPT$.
Section~\ref{sec:line-warmup} proves these interval inequalities and describes the automatic implementation.

\subsection{Searching for hidden requests}
\label{subsec:overview-blind}

On a known finite line, each blind arrival reveals a request's label and arrival time; its location is learned only on a visit.
Labels identify when all announced requests have been found.
The server checks its current position for free and serves arrivals there at age zero.
Let $\delta>0$ denote the minimum distance between distinct sites.
A one-site input costs zero.

\paragraph{Waiting and searching.}
All searches in a phase start and end at the anchor $s_{h-1}$.
Under elective service, when the oldest phase request reaches age $r_j=2^j\delta$, for $j=0,1,\ldots$, the server sweeps the sites in $[s_{h-1}-r_j,s_{h-1}+r_j]$ without serving requests and returns to the anchor.
The sweep costs at most $4r_j$.
After an unsuccessful search, the server waits for the next radius.
At the first successful radius, one additional walk serves all pending requests and ends at a farthest request $p_h$.
Search and service occur at the same time, with no arrival between these walks.

Let $R_h=|s_h-s_{h-1}|$ denote the distance to that farthest request, and let $r$ denote the first successful radius.
Success gives $R_h\le r$.
Failure at the preceding radius gives $r\le2R_h$; at the initial radius this follows from the minimum site spacing.
The requests missed by an earlier search remain in the phase, even if further requests arrive.
The radii sum to less than $2r$, giving total search movement less than $8r$.
The final service walk visits the nearer extreme first and ends at the farther one, costing at most $3R_h$; its delay is $r$.
Thus
\begin{equation}
\label{eq:overview-blind-cost}
\ALG_h<8r+3R_h+r\le21R_h.
\end{equation}
The bound consists of $16R_h$ for searching, $3R_h$ for the service walk, and $2R_h$ for delay.

\paragraph{The certificate.}
The group is now just the two endpoint requests $G_h=\{p_{h-1},p_h\}$, and its certificate is $C_h=R_h$.
Their Steiner cost is their distance $R_h$.
The phase lasts at least $r\ge R_h$, and the two requests arrive within $I_h$.
The odd and even groups satisfy the same disjointness conditions as before.
Corollary~\ref{cor:certificate-transfer} therefore gives $\ALG\le21\sum_hR_h\le84\OPT$.
The algorithm never needs to know $R_h$ before the successful search; it is used only in the analysis.

\paragraph{Automatic service.}
A search now serves every request it discovers, potentially creating several delay charges in one phase.
The algorithm retains a virtual record of each served request until the phase ends, preserving the same search times and final anchor.
A sweep of radius $r_j$ costs at most $4r_j$ in movement and $r_j$ in delay.
Once all requests have been found, the server moves directly to the farthest recorded request, costing $R_h$.
The phase cost is less than $10r+R_h\le21R_h$, giving the same ratio $84$.
Section~\ref{sec:blind-line} specifies the feedback and proves these claims for both service rules.

The algorithms above use the same certificate comparison, but different geometry determines the thresholds, walks, and phase cost bounds.
The remaining sections establish the detailed guarantees, exact offline algorithms, and lower bounds.
The online pseudocode appears with the analysis in Sections~\ref{sec:visible-details} and~\ref{sec:blind-line}; Appendix~\ref{app:algorithms} gives the offline pseudocode.

\FloatBarrier
\section{Further Analysis of Visible Algorithms}
\label{sec:visible-details}

Section~\ref{sec:online-overview} gives the algorithms and their certificate groups.
We now verify the interval inequalities used on a line, state the tree guarantee, and extend the general-metric argument to any approximation of a connecting tree.
The common arrival and timer conventions are those of Section~\ref{subsec:execution-conventions}.
All ties in the pseudocode are resolved by fixed orders on sites, edges, and request labels.
Let $Q$ denote the total number of requests.

\paragraph{Requests served at arrival.}
For the automatic implementations, an algorithm may need to remember a request that has already been served at age zero.
The following observation justifies keeping such records under elective service as well.
\begin{lemma}
 \label{lem:zero-age-arrivals}
Under elective service, an online algorithm may, without loss of generality, immediately serve every request that arrives at the current server position.
\end{lemma}

\begin{proof}
The modified algorithm serves such a request immediately and keeps a record of it, allowing the original algorithm to continue as if it were pending.
We call this record a \emph{virtual request}.
The modified algorithm follows the original walks and omits the real request when its planned batch is served.
The movement is unchanged, and removing a request from a batch cannot increase that batch's maximum delay.
\end{proof}

\subsection{Finite lines}
\label{sec:line-warmup}
\label{sec:global-hull}
\label{subsec:metric-representations}

Let $X=\{x_1<\cdots<x_N\}\subset\mathbb{R}$ denote the sites, with $d(x,y)=|x-y|$.
A walk specified by its turning points visits every site between consecutive turning points, without increasing its length.
We use the interval algorithm from Section~\ref{subsec:overview-line}; its pseudocode is Algorithm~\ref{alg:visible-global-hull}.
\paragraph{Notation.}

The initial values are $\tau_0=0$ and $H_0=\{s_0\}$.
Suppose phase $h-1$ served all pending requests at time $\tau_{h-1}$.
Let $H_{h-1}$ denote the smallest interval containing their locations; the server ends at one of its endpoints.
At a later time $t$, let $H(t)$ denote the smallest interval containing the pending locations, and define
\begin{equation}
\label{eq:global-hull-budget}b(t)=\operatorname{diam}(H(t))+\dist(H_{h-1},H(t)).
\end{equation}
If no request is pending, the algorithm waits.
Otherwise, it triggers at the first time at which the oldest pending waiting time is at least $b(t)$.
It then serves every pending request by a shortest walk that first reaches the nearer endpoint of $H(t)$, covers the interval, and ends at the opposite endpoint.
Let $B_h$, $H_h$, and $D_h$ denote the batch served, its interval, and the interval length, respectively.
Let $b_h$ denote the waiting threshold and $e_h$ denote the final server position.

\begin{algorithm}[!htb]
 \caption{Visible global hull service on a line}
 \label{alg:visible-global-hull}
 \begin{algorithmic}[1]
    \Require Finite line $X$, initial position $s_0$; elective service
    \State $e\gets s_0$; $H_{\mathrm{prev}}\gets\{s_0\}$
    \Loop
      \State Wait for the first complete epoch leaving a nonempty pending set $V$
      \State $\alpha\gets\min_{q\in V}a(q)$
      \Loop
        \State $H=[L,R]\gets\operatorname{hull}\{x(q):q\in V\}$
        \State $b\gets(R-L)+\dist(H_{\mathrm{prev}},H)$; $c\gets\alpha+b$\label{line:line-threshold}
        \State $u\gets$ time of the next arrival epoch, or $+\infty$
        \If{$u\le c$}\label{line:line-arrivals}
          \State Advance to $u$, process the complete epoch, and append it to $V$
          \State \textbf{continue}
        \EndIf
        \State Advance to $c$ and \textbf{break}
      \EndLoop
      \If{$|e-L|\le |e-R|$}\label{line:line-nearer}
        \State $P\gets(e,L,R)$; $e\gets R$
      \Else
        \State $P\gets(e,R,L)$; $e\gets L$
      \EndIf
      \State Traverse $P$ and electively serve every request in $V$
      \State $H_{\mathrm{prev}}\gets H$; clear $V$
    \EndLoop
  \end{algorithmic}
\end{algorithm}

\paragraph{Threshold and certificate.}
\label{subsec:line-online-analysis}

We first show that the algorithm serves each batch when its oldest request has waited exactly the threshold in \eqref{eq:global-hull-budget}.
We then bound the phase cost using the smallest interval containing the current and preceding batches.

On a line, the quantity in~\eqref{eq:global-hull-budget} is nondecreasing as the pending hull expands.
Indeed, for $H_{h-1}=[a,b]$ and $H(t)=[\ell,r]$ it equals
\begin{equation}
\label{eq:online-global-hull-display-1}r-\ell+\max\{a-r,\ell-b,0\}.
\end{equation}
This equals $\max\{a-\ell,r-b,r-\ell\}$, whose three terms are nondecreasing as $\ell$ decreases and $r$ increases.
Arrivals at a trigger timestamp are processed before the trigger test (line~\ref{line:line-arrivals} in Algorithm~\ref{alg:visible-global-hull}).
Thus an arrival can only raise the threshold, and between arrivals the oldest waiting time increases continuously while the threshold is fixed.
The first successful crossing therefore occurs with equality, implying that the batch delay charge of phase $h$ is $b_h$.
After the final input arrival, such a crossing occurs in finite time.
Moreover, every request of $B_h$ arrived no earlier than $\tau_{h-1}$, and therefore
\begin{equation}
\label{eq:phase-duration-budget}\tau_h-\tau_{h-1}\ge b_h.
\end{equation}

For the transition from $H_{h-1}$ to $H_h$, let $X_h$ denote the span of both intervals:
\begin{equation}
\label{eq:transition-span}X_h=\operatorname{diam}(H_{h-1}\cup H_h).
\end{equation}
The group for certificate $X_h$ consists of requests at the two endpoints of this interval.
For $h=1$, one extreme may be the dummy request of zero cost at $(s_0,0)$.
Theorem~\ref{thm:certificate-engine} applies to these pairs because the Steiner cost of two points is their distance.

Each phase costs at most $5X_h/2$.
We prove below that the request pairs for odd phases and those for even phases each satisfy the certificate theorem.
\begin{theorem}
 \label{thm:global-hull-competitive}
On a finite line with visible requests and elective service, Algorithm~\ref{alg:visible-global-hull} is deterministic and has competitive ratio at most $10$ under the model of Section~\ref{sec:model}.
\end{theorem}

\begin{proof}
Let $L_h,R_h$ denote the endpoints of $H_h$.
The nearer-endpoint choice (line~\ref{line:line-nearer} in Algorithm~\ref{alg:visible-global-hull}) gives a phase walk of length
\begin{equation}
\label{eq:online-global-hull-display-2}K_h=D_h+\min\{|e_{h-1}-L_h|,|e_{h-1}-R_h|\}.
\end{equation}
If $e_{h-1}\notin H_h$, the server moves in one direction throughout the walk, implying $K_h\le X_h$.
Otherwise $K_h\le3D_h/2\le3X_h/2$.
The delay satisfies
\begin{equation}
\label{eq:online-global-hull-display-3}b_h=D_h+\dist(H_{h-1},H_h)\le X_h.
\end{equation}
Hence the phase cost is at most $5X_h/2$.

After omission of a possible zero-cost phase at time zero, completion times are strictly increasing.
The odd and even phases use the windows from Section~\ref{subsec:overview-connecting-trees}, namely,
\begin{equation}
\label{eq:online-global-hull-display-5}W_h=(\tau_{h-2},\tau_h],
\end{equation}
for $h\ge3$, with $[0,\tau_1]$ and $[0,\tau_2]$ initially.
The pair for phase $h$ uses requests from $B_{h-1}\cup B_h$.
Thus each pair arrives in its window, and no request is used twice within one color.
When $X_h=0$, a singleton group suffices.
Writing $g_h=\dist(H_{h-1},H_h)$, equation~\eqref{eq:phase-duration-budget} gives, for $h\ge2$,
\begin{align}
\label{eq:online-global-hull-display-7}\tau_h-\tau_{h-2}\ge b_{h-1}+b_h\ge D_{h-1}+g_h+D_h\ge X_h.
\end{align}
Initially, $\tau_1\ge b_1=X_1$.
The two endpoints are at distance $X_h$, meaning that their Steiner cost is $X_h$.
Corollary~\ref{cor:certificate-transfer}, together with the phase cost bound, yields
\begin{equation}
\label{eq:online-global-hull-display-6}\ALG\le\frac52\sum_hX_h\le10\OPT(I),
\end{equation}
which proves the theorem.
\end{proof}

The sharpness example for the certificate bound appears in Appendix~\ref{app:supplementary}.
The line algorithm also admits the following automatic implementation.
\begin{corollary}
 \label{cor:visible-automatic-competitive}
On a finite line with visible requests and automatic service, Algorithm~\ref{alg:visible-global-hull} has a deterministic implementation with competitive ratio at most $10$ under the model of Section~\ref{sec:model}.
\end{corollary}

\begin{proof}
The automatic algorithm follows the elective execution and retains virtual copies of requests served at the parked point at age zero.
Each trigger walk then has identical movement and no greater delay.
Theorem~\ref{thm:global-hull-competitive} and the equality of offline values in Theorem~\ref{thm:auto-elective-equivalence} give the result.
\end{proof}

\subsection{Weighted trees}
\label{sec:tree-extension}
\label{subsec:tree-online-description}

For a finite tree $T=(V,E)$ with positive edge weights, let $T[A]$ denote the minimal subtree connecting $A\subseteq V$ and let $w$ denote total edge weight.
Thus $\ST_V(A)=w(T[A])$.
During a phase with anchor $s$, the threshold from Section~\ref{subsec:overview-connecting-trees} is
\begin{equation}
\label{eq:tree-extension-display-2}
W(t)=w\bigl(T[\{s\}\cup\{x(q):q\in P(t)\}]\bigr),
\end{equation}
where $P(t)$ contains the phase requests announced by time $t$.
The server uses an open depth first traversal: each edge on the path to the chosen endpoint is used once, and each other edge twice.
The endpoint is an actual phase request, for example one farthest from $s$ in the subtree.
Algorithm~\ref{alg:tree-subtree-trigger} gives the implementation.

\begin{algorithm}[!htb]
 \caption{Visible subtree threshold service on a weighted tree}
 \label{alg:tree-subtree-trigger}
 \begin{algorithmic}[1]
    \Require Weighted tree $T$, initial anchor $s_0$, service semantics
    \State $s\gets s_0$
    \Loop
      \State Serve and discard every arrival at the idle anchor $s$ at zero cost
      \State Wait for the first complete epoch containing an arrival away from the anchor
      \State $P\gets$ its requests away from the anchor; $\alpha\gets\min_{q\in P}a(q)$
      \Loop
        \State $F\gets T[\{s\}\cup\{x(q):q\in P\}]$; $W\gets w(F)$
        \State $c\gets\alpha+W$; $u\gets$ next arrival time, or $+\infty$\label{line:tree-trigger}
        \If{$u\le c$}
          \State Advance to $u$ and process the complete epoch
          \State Serve/discard its anchor arrivals and append all others to $P$
          \State \textbf{continue}
        \EndIf
        \State Advance to $c$ and \textbf{break}
      \EndLoop
      \State Choose an actual phase request $p$ farthest from $s$ in $F$\label{line:tree-endpoint}
      \State Use an open depth first traversal of $F$, with the path to $x(p)$ last; serve $P$ in one event\label{line:tree-walk}
      \State $s\gets x(p)$; clear $P$
    \EndLoop
  \end{algorithmic}
\end{algorithm}

\paragraph{Competitive guarantee.}
\label{subsec:tree-online-analysis}
\begin{theorem}
\label{thm:visible-tree-competitive}
On a finite weighted tree, Algorithm~\ref{alg:tree-subtree-trigger} is deterministic, runs in polynomial time, and has competitive ratio at most $12$ under either service rule.
\end{theorem}
\begin{proof}
The threshold is nondecreasing because the phase terminal set only grows.
Every numbered phase contains a request away from the anchor and therefore has positive threshold.
The timer in line~\ref{line:tree-trigger} of Algorithm~\ref{alg:tree-subtree-trigger} gives delay $W_h$ at completion; the traversal in line~\ref{line:tree-walk} costs at most $2W_h$.
The endpoint choice (line~\ref{line:tree-endpoint}) identifies the request $p_h$ used in the next group.
The groups $G_h=\{p_{h-1}\}\cup P_h$ and the windows of Section~\ref{subsec:overview-connecting-trees} satisfy $W_h=\ST_V(\{x(q):q\in G_h\})$ and $|I_h|\ge W_h$.
Corollary~\ref{cor:certificate-transfer} yields
\begin{equation}
\label{eq:tree-extension-display-3}
\ALG\le3\sum_h W_h\le12\OPT.
\end{equation}
Arrivals at the anchor are served at age zero, and the final traversal clears every remaining phase request under either rule.
Theorem~\ref{thm:auto-elective-equivalence} gives the same bound against the automatic optimum.
After the last arrival, the fixed threshold is eventually reached.
There are at most $Q$ phases and arrival epochs, where $Q$ is the number of requests; subtree computations and traversals take polynomial time.
\end{proof}

\subsection{General metrics and approximate connecting trees}
\label{sec:general-metric}

The analysis in Section~\ref{subsec:overview-connecting-trees} uses an exact tree on a weighted tree metric and a terminal MST on a general metric.
Here we state the oracle guarantee, prove the extension to any connecting-tree approximation, and verify the MST and automatic implementations.
All algorithms use the phase requests and endpoint labels from Section~\ref{subsec:overview-phases}.

\paragraph{Exact Steiner trees.}
\label{subsec:steiner-description}
\label{subsec:steiner-analysis}
For the current anchor $s$ and phase requests $P(t)$, the waiting threshold is
\begin{equation}
\label{eq:general-metric-steiner-threshold}
S(t)=\ST_X(\{s\}\cup\{x(q):q\in P(t)\}).
\end{equation}
At the trigger, an oracle returns a tree $F_h$ attaining $S_h$.
The server follows an open depth first traversal and ends at an actual phase request.
Algorithm~\ref{alg:exact-steiner-phase} gives the pseudocode.
\begin{algorithm}[!htb]
 \caption{Visible exact Steiner phase algorithm}
 \label{alg:exact-steiner-phase}
 \begin{algorithmic}[1]
    \Require Finite metric $(X,d)$, exact Steiner oracle, start $s_0$
    \State $s\gets s_0$
    \Loop
      \State Serve and discard arrivals at the idle anchor at zero cost
      \State Wait for the first complete epoch containing an arrival away from the anchor
      \State $P\gets$ its requests away from the anchor; $\alpha\gets\min_{q\in P}a(q)$
      \Loop
        \State $F\gets$ the canonical exact Steiner tree for
          $\{s\}\cup\{x(q):q\in P\}$
        \State $S\gets w(F)$
        \State $c\gets\alpha+S$; $u\gets$ next arrival time, or $+\infty$\label{line:steiner-trigger}
        \If{$u\le c$}
          \State Advance to $u$ and process the complete epoch
          \State Serve/discard its anchor arrivals and append all others to $P$
          \State \textbf{continue}
        \EndIf
        \State Advance to $c$ and \textbf{break}
      \EndLoop
      \State Choose an actual phase request $p$ farthest from $s$ in tree $F$
      \State Use an open depth first traversal of $F$, with the path to $x(p)$ last; serve $P$ in one event\label{line:steiner-walk}
      \State $s\gets x(p)$; clear $P$
    \EndLoop
  \end{algorithmic}
\end{algorithm}

\begin{theorem}
\label{thm:general-metric-steiner-competitive}
On a finite metric with visible requests and elective service, Algorithm~\ref{alg:exact-steiner-phase} has competitive ratio at most $12$, provided an exact Steiner tree oracle is available.
\end{theorem}
\begin{proof}
Adding terminals cannot decrease $S(t)$.
The timer in line~\ref{line:steiner-trigger} of Algorithm~\ref{alg:exact-steiner-phase} therefore gives age exactly $S_h$ at service, and the phase lasts at least $S_h$.
The traversal in line~\ref{line:steiner-walk} has length $2S_h-d_{F_h}(s,s_h)\le2S_h$, where $d_{F_h}$ denotes distance in the returned tree.
The groups $G_h=\{p_{h-1}\}\cup P_h$ have Steiner cost $S_h$ and lie in the windows from Section~\ref{subsec:overview-connecting-trees}.
Their phase costs are at most $3S_h$, giving ratio $12$ by Corollary~\ref{cor:certificate-transfer}.
Finite inputs terminate because the threshold is fixed after the last arrival.
\end{proof}

\paragraph{Approximate trees.}
\label{subsec:surrogate-transfer}
An approximate connecting tree may become lighter when a terminal is added.
The algorithm therefore retains the largest weight observed during the phase.
Let $(X,d)$ denote a finite metric and let $\rho\ge1$ denote an approximation factor.
Suppose that for every nonempty terminal set $A\subseteq X$ a deterministic procedure returns a tree $F(A)$ in the complete metric spanning $A$, with weight $\Gamma(A)$ satisfying
 \begin{equation}
\label{eq:spatial-surrogate-factor}\Gamma(A)\le\rho\ST_X(A).
\end{equation}
The phase convention is that of Section~\ref{subsec:overview-phases}; arrivals at the anchor are served and omitted at age zero.
Within each phase, the retained weight and waiting threshold are
 \begin{equation}
\label{eq:general-metric-display-7}E(t)=\max_{u\le t}\Gamma(A(u)),\qquad b(t)=E(t)/\rho,
\end{equation}
where the maximum ranges over prefix times of the current phase and $A(u)$ consists of its anchor and terminals released by time $u$.
The server triggers when the oldest phase request has age $b(t)$ and makes an open depth first traversal of $F(A(t))$ in one service event, ending at a deterministically chosen actual phase request.
All ties in the tree procedure and endpoint choice are resolved canonically.
\begin{proposition}
 \label{prop:spatial-surrogate-transfer}
This algorithm is deterministic and has competitive ratio at most $4(2\rho+1)$ for visible elective service.
If $\rho$ is rational and the tree procedure runs in polynomial time, the phase algorithm also runs in polynomial time.
\end{proposition}

\begin{proof}
By definition, $E(t)$ and hence $b(t)$ do not decrease during a phase.
The algorithm therefore serves at age exactly $b(t)$, and after the last arrival every phase ends in finite time.
For every prefix of the final terminal set $A_h$,
\begin{equation}
\label{eq:general-metric-display-8}\Gamma(A(u))\le\rho\ST_X(A(u))\le\rho\ST_X(A_h).
\end{equation}
Taking the maximum gives $b_h=E_h/\rho\le\ST_X(A_h)$.
The phase lasts at least $b_h$, and the current tree has weight at most $E_h=\rho b_h$.
Its traversal and delay therefore cost at most $(2\rho+1)b_h$.

The groups are $G_h=\{p_{h-1}\}\cup P_h$, with the windows of Section~\ref{subsec:overview-connecting-trees} and certificates $b_h$.
The preceding bounds verify Corollary~\ref{cor:certificate-transfer}, giving
\begin{equation}
\label{eq:general-metric-display-9}\ALG\le(2\rho+1)\sum_hb_h\le4(2\rho+1)\OPT.
\end{equation}
There are at most $Q$ phases and epochs.
Updating the maximum and traversing the returned trees adds only polynomial work.
\end{proof}

\paragraph{Terminal MSTs.}
\label{subsec:mst-description}
\label{subsec:mst-analysis}
For the polynomial time algorithm, let $M_h(t)=\MST(A(t))$ denote the current terminal MST weight.
Its threshold is
\begin{equation}
\label{eq:mst-monotone-envelope}
E_h(t)=\max_{\tau_{h-1}\le u\le t}M_h(u),\qquad b_h(t)=E_h(t)/2,
\end{equation}
where the maximum is over prefixes of the current phase and is reset at the start of each phase.
Algorithm~\ref{alg:mst-envelope} updates the MST after each arrival epoch and traverses the current tree at the trigger.
Figure~\ref{fig:mst-envelope} illustrates why the maximum is retained.

\begin{algorithm}[!htb]
 \caption{Service using the running maximum of terminal MST weights}
 \label{alg:mst-envelope}
 \begin{algorithmic}[1]
    \Require Finite metric $(X,d)$, initial anchor $s_0$
    \State $s\gets s_0$
    \Loop
      \State Serve and discard arrivals at the idle anchor at zero cost
      \State Wait for the first complete epoch containing an arrival away from the anchor
      \State $P\gets$ its requests away from the anchor; $\alpha\gets\min_{q\in P}a(q)$; $E\gets0$
      \Loop
        \State $F\gets$ the canonical minimum spanning tree on $\{s\}\cup\{x(q):q\in P\}$
        \State $M\gets w(F)$
        \State $E\gets\max\{E,M\}$; $b\gets E/2$\label{line:mst-envelope}
        \State $c\gets\alpha+b$; $u\gets$ next arrival time, or $+\infty$
        \If{$u\le c$}
          \State Advance to $u$ and process the complete epoch
          \State Serve/discard its anchor arrivals and append all others to $P$
          \State \textbf{continue}
        \EndIf
        \State Advance to $c$ and \textbf{break}
      \EndLoop
      \State Choose a canonical actual phase request $p$ away from the anchor in current $F$
      \State Use an open depth first traversal of $F$, with the path to $x(p)$ last; serve $P$ in one event\label{line:mst-walk}
      \State $s\gets x(p)$; clear $P$
    \EndLoop
  \end{algorithmic}
\end{algorithm}

\begin{lemma}
 \label{lem:terminal-mst-steiner-factor-two}
For every nonempty terminal set $A\subseteq X$,
 \begin{equation}
\label{eq:general-metric-display-10}\MST(A)\le2\ST_X(A).
\end{equation}
\end{lemma}

\begin{proof}
A doubled optimal Steiner tree has an Euler tour of twice its weight.
Shortcutting this tour to the terminals in first-visit order and omitting the closing edge gives a spanning path.
The resulting terminal path is a spanning tree of weight at most $2\ST_X(A)$ by the triangle inequality; therefore, the MST is no heavier.
For a singleton, both weights are zero.
\end{proof}

\begin{theorem}
\label{thm:general-metric-mst-competitive}
On a finite metric with visible requests and elective service, Algorithm~\ref{alg:mst-envelope} is deterministic, runs in polynomial time, and has competitive ratio at most $20$.
\end{theorem}
\begin{proof}
Lemma~\ref{lem:terminal-mst-steiner-factor-two} verifies Proposition~\ref{prop:spatial-surrogate-transfer} with $\Gamma=\MST$ and $\rho=2$.
Algorithm~\ref{alg:mst-envelope} implements its retained maximum and threshold in line~\ref{line:mst-envelope}, and its traversal in line~\ref{line:mst-walk}.
The phase estimates give $b_h\le\ST_X(\{x(q):q\in G_h\})$, $b_h\le|I_h|$, and phase cost at most $5b_h$.
Hence $\ALG\le5\sum_h b_h\le20\OPT$.
The algorithm recomputes an MST after each epoch and traverses it once per phase, taking $O(Q\,\mathrm{poly}(|X|))$ time.
Repeated locations occupy one MST vertex, while request identities remain distinct.
\end{proof}

\paragraph{Automatic service.}
\label{subsec:metric-automatic-analysis}
\begin{corollary}
\label{cor:general-metric-automatic}
On a finite metric with visible requests and automatic service, Algorithms~\ref{alg:exact-steiner-phase} and~\ref{alg:mst-envelope} have deterministic implementations with competitive ratios $12$ and $20$, respectively.
The first requires an exact Steiner oracle, and the second runs in polynomial time.
\end{corollary}
\begin{proof}
Arrivals at the anchor are served at age zero; including or omitting their locations does not change the terminal set, since it already contains the anchor.
The server waits there until the trigger and then clears all remaining phase requests in one event.
Its movement, delay bound, and final endpoint are unchanged.
The two preceding theorems and Theorem~\ref{thm:auto-elective-equivalence} therefore give the stated guarantees.
\end{proof}

\section{Blind Service on Finite Lines}
\label{sec:blind-line}

When locations are hidden, the server must search for requests.
Section~\ref{subsec:overview-blind} gives the search algorithm and its certificate argument.
Here we specify its feedback, verify the phase cost, and describe the automatic implementation.
Algorithms~\ref{alg:blind-elective-dyadic} and~\ref{alg:blind-automatic-dyadic} give the pseudocode; the lower bounds appear in Section~\ref{subsec:blind-lower-bounds}.

\paragraph{Feedback and walks.}
\label{subsec:information-models}

An arrival announcement gives a request's label and arrival time, but not its location.
The algorithm can compute its waiting time and match each discovery to an announcement.
It can therefore determine when every announced request has been found.
Equivalent feedback giving arrival times and multiplicities also suffices if it provides this information.
A visit reveals every pending request at that location, including its label and arrival time.
The server can check its current position without moving, including for new arrivals there.
Under elective service, discovery does not force service; under automatic service, it does.
The offline optimum knows all locations.

The known line has sites $X=\{x_1<\cdots<x_N\}$; walks visit every listed site between consecutive turning points.
For $N\ge2$, let $\delta=\min_{x\ne y}|x-y|$ denote the minimum site spacing.
When $N=1$, every request is served at arrival at zero cost.
A \emph{sweep} visits every site of an interval and returns to its starting point, at movement cost at most twice the interval length.
This cost does not depend on the number of sites in the interval.

The event order from Section~\ref{sec:model} also applies here.
After processing all arrivals at a timestamp, the server may make several walks and choose each using information learned on the preceding walks.
We call this sequence an \emph{atomic action}: no arrival occurs between its walks, and each walk is charged as a separate event.
Thus a successful search and a following service or repositioning walk can occur at the same time.
Figure~\ref{fig:visible-blind-space-time} compares the visible and blind models on one input.

\begin{figure}[ht]
  \centering
  \begin{tikzpicture}[
      x=0.56cm,
      y=0.82cm,
      >=Latex,
      axis/.style={draw=black!65,thin,-{Latex[length=2mm]}},
      server/.style={draw=black,very thick,-{Latex[length=2mm]}},
      wait/.style={draw=black!70,thick},
      hidden/.style={draw=orange!85!black,densely dotted,thick},
      known/.style={draw=blue!65!black,dashed,thick},
      probe/.style={draw=blue!65!black,densely dashed,thick,{Latex[length=1.7mm]}-{Latex[length=1.7mm]}},
      every node/.style={font=\scriptsize}
    ]

    \begin{scope}[xshift=0cm]
      \node[font=\small\bfseries] at (0,4.75) {(a) Visible locations};
      \draw[axis] (-4.35,0) -- (4.45,0) node[right] {$x$};
      \draw[axis] (-4.2,-0.12) -- (-4.2,4.55) node[above] {$t$};
      \foreach \x/\lab in {-1/$-1$,0/$0$,3/$3$}
        \draw (\x,0.06)--(\x,-0.06) node[below=2pt] {\lab};
      \foreach \t in {1,2,4}
        \draw (-4.14,\t)--(-4.26,\t) node[left=1pt] {\t};

      \fill[black] (0,0) rectangle +(0.13,0.13);
      \node[below right=2pt] at (0,0) {$s_0$};
      \fill[blue!65!black] (-1,0) circle (2.1pt)
        node[above left=1pt] {$q_1$};
      \fill[blue!65!black] (3,0.6) circle (2.1pt)
        node[right=2pt] {$q_2$};
      \draw[known] (-1,0.06)--(-1,4);
      \draw[known] (3,0.66)--(3,4);
      \draw[wait] (0,0.14)--(0,4);
      \node[anchor=west,text=black!70] at (0.18,2.8)
        {known interval; wait until $t=4$};
      \draw[server] (0,4)--(-1,4);
      \draw[server] (-1,4)--(3,4);
      \node[above=2pt,align=center] at (1.15,4)
        {one batch covers $[-1,3]$};
    \end{scope}

    \begin{scope}[xshift=8.5cm]
      \node[font=\small\bfseries] at (0,4.75) {(b) Blind locations};
      \draw[axis] (-4.35,0) -- (4.45,0) node[right] {$x$};
      \draw[axis] (-4.2,-0.12) -- (-4.2,4.55) node[above] {$t$};
      \foreach \x/\lab in {-1/$-1$,0/$0$,3/$3$}
        \draw (\x,0.06)--(\x,-0.06) node[below=2pt] {\lab};
      \foreach \t in {1,2,4}
        \draw (-4.14,\t)--(-4.26,\t) node[left=1pt] {\t};

      \fill[black] (0,0) rectangle +(0.13,0.13);
      \node[below right=2pt] at (0,0) {$s_0$};
      \fill[orange!85!black] (-1,0) circle (2.1pt)
        node[above left=1pt] {$q_1$};
      \fill[orange!85!black] (3,0.6) circle (2.1pt)
        node[right=2pt] {$q_2$};
      \draw[hidden] (-1,0.06)--(-1,1);
      \draw[hidden] (3,0.66)--(3,4);
      \draw[wait] (0,0.14)--(0,4);

      \draw[probe] (-1,1)--(1,1);
      \node[above right=1pt,text=blue!65!black] at (0.15,1) {$r=1$};
      \draw[fill=white,draw=blue!65!black,thick] (-1,1) circle (2.2pt);
      \node[above left=1pt,text=blue!65!black] at (-1,1) {discover $q_1$};
      \draw[known] (-1,1.06)--(-1,4);

      \draw[probe] (-2,2)--(2,2);
      \node[above right=1pt,text=blue!65!black] at (0.15,2) {$r=2$};

      \draw[probe] (-4,4)--(4,4);
      \node[above left=1pt,text=blue!65!black] at (4,4) {$r=4$: all found};
      \draw[fill=white,draw=blue!65!black,thick] (3,4) circle (2.2pt);
      \draw[server] (0,4)--(-1,4);
      \draw[server] (-1,4)--(3,4);
      \node[below=3pt,align=center,fill=white,inner sep=1pt] at (1.2,4)
        {serve both at time $4$};
    \end{scope}
  \end{tikzpicture}
  \caption{A space--time view of the same two-request instance on
  $X=\{-4,-3,\ldots,4\}$.  The server starts at $0$; requests
  $q_1=(-1,0)$ and $q_2=(3,0.6)$ are shown as (location, release-time)
  points, and time runs upward.  In the visible model, their locations and
  interval are known at arrival. The server waits until time $4$
  and serves both in one walk. In the blind elective model, dashed
  horizontal bars show the intervals searched at radii $1$, $2$, and $4$.
  Each search returns to the starting point without serving requests.
  The first search finds $q_1$; the radius-$4$ search finds $q_2$.
  Having found both announced requests, the server serves them in one
  batch at time $4$. Dotted
  vertical segments represent a request whose location is still hidden;
  dashed vertical segments represent a discovered or visible request that
  remains unserved.}
  \label{fig:visible-blind-space-time}
\end{figure}

\subsection{Elective service}
\label{subsec:blind-elective-description}

Algorithm~\ref{alg:blind-elective-dyadic} waits between searches and doubles the search radius after each unsuccessful search.
All searches in a phase start and end at the same point, its anchor.
Once every pending request has been found, one final walk serves them as a batch.

The algorithm operates in phases of positive cost.
Phase $k$ starts with no pending request and with the server at an anchor $s_{k-1}$.
Requests arriving at the anchor are observed and served immediately at zero cost.
They do not start or end a phase of positive cost.
Let $\alpha_k$ denote the arrival time of the oldest remaining pending request.
At ages
\begin{equation}
\label{eq:blind-service-display-5}r_j=2^j\delta,\qquad j=0,1,2,\ldots,
\end{equation}
the server makes a closed exhaustive sweep of $X\cap[s_{k-1}-r_j,s_{k-1}+r_j]$, returning to the anchor and deliberately serving no request.
At the first successful radius, the server immediately follows a shortest walk serving all pending requests and ending at a farthest pending location from the anchor.
This walk belongs to the same atomic action as the successful sweep.
Let $p_k$ denote this terminal request, $s_k$ its location, and $T_k$ the completion time.
Let $p_0$ denote a dummy request at $(s_0,0)$; the initial completion time is $T_0=0$.

\begin{algorithm}[!htb]
 \caption{Blind elective service with doubling search radii}
 \label{alg:blind-elective-dyadic}
 \begin{algorithmic}[1]
    \Require Known finite line $X$ with $|X|\ge2$, initial anchor $s_0$
    \State $\delta\gets\min\{|x-y|:x,y\in X,\ x\ne y\}$; $s\gets s_0$
    \Loop
      \State Wait for the next complete arrival epoch and process it
      \State Check $s$ without moving and serve every identity revealed there
      \State $P\gets$ identities from the epoch unresolved after the check at the anchor
      \If{$P$ is empty} \State \textbf{continue} \EndIf
      \State $\alpha\gets\min_{q\in P}a(q)$; $j\gets0$
      \Loop
        \State $r\gets2^j\delta$; $c\gets\alpha+r$; $u\gets$ next arrival time, or $+\infty$
        \If{$u\le c$}
          \State Advance to $u$, process its whole epoch, and check the anchor $s$
          \State Append every identity unresolved after the check at the anchor to $P$
          \If{$u<c$} \State \textbf{continue} \EndIf
        \Else
          \State Advance to $c$
        \EndIf
        \State Sweep $s\to L_r\to R_r\to s$ in one component without serving\label{line:blind-elective-sweep}
        \Statex \Comment{$L_r,R_r$ are the extreme sites of $X\cap[s-r,s+r]$}
        \State Mark every encountered identity and location as discovered
        \If{some identity in $P$ is undiscovered}\label{line:blind-elective-failure}
          \State $j\gets j+1$; \textbf{continue}
        \EndIf
        \State Let $L,R$ denote the extreme discovered locations in $P$
        \If{$|s-R|\ge|s-L|$}\label{line:blind-elective-cleanup}
          \State Traverse $s\to L\to R$ and set $s\gets R$
        \Else
          \State Traverse $s\to R\to L$ and set $s\gets L$
        \EndIf
        \State Electively serve all identities in $P$ in this cleanup component
        \State Clear $P$ and \textbf{break}
      \EndLoop
    \EndLoop
  \end{algorithmic}
\end{algorithm}

\paragraph{Phase cost and competitive ratio.}
\label{subsec:blind-elective-analysis}

Let $R_k$ denote the distance from the anchor to the farthest request in phase $k$.
Although the algorithm does not know $R_k$ while searching, its first successful radius is at most $2R_k$.
We use this to bound the phase cost.
\begin{lemma}
 \label{lem:blind-phase-bound}
For each phase $k$ of Algorithm~\ref{alg:blind-elective-dyadic}, let $P_k$ denote its request set and define
 \begin{equation}
\label{eq:blind-service-display-6}R_k=|s_k-s_{k-1}|=\max_{q\in P_k}|x(q)-s_{k-1}|.
\end{equation}
The cost of phase $k$ is less than $21R_k$.
\end{lemma}

\begin{proof}
For the first successful radius $r=2^j\delta$, success gives $R_k\le r$.
If $j>0$, the failed search at $r/2$ (line~\ref{line:blind-elective-failure} in Algorithm~\ref{alg:blind-elective-dyadic}) leaves a request beyond $r/2$ in the phase.
If $j=0$, minimal spacing gives $R_k\ge\delta$.
Hence
\begin{equation}
\label{eq:blind-final-scale}R_k\le r\le2R_k.
\end{equation}
Each sweep in line~\ref{line:blind-elective-sweep} costs at most four times its radius, giving total exploration cost
\begin{equation}
\label{eq:blind-service-display-7}4\sum_{i=0}^j2^i\delta<8r.
\end{equation}
The cleanup visits the nearer extreme first (line~\ref{line:blind-elective-cleanup}).
For the nonnegative left and right extents $a,b$ of the pending locations about the anchor, cleanup costs at most
\begin{equation}
\label{eq:blind-service-display-8}2\min\{a,b\}+\max\{a,b\}\le3R_k.
\end{equation}
Its delay is exactly $r$, since the oldest request remains pending and cleanup follows the sweep in the same atomic action.
The total is therefore less than
\begin{equation}
\label{eq:blind-service-display-9}8r+3R_k+r\le21R_k,
\end{equation}
which proves the lemma.
\end{proof}

We use the requests at two successive anchors as a certificate group.
Their distance is $R_k$. The preceding phase bound therefore gives the following guarantee.
\begin{theorem}
 \label{thm:blind-elective-competitive}
On a known finite line with blind requests and elective service, Algorithm~\ref{alg:blind-elective-dyadic} is deterministic and has competitive ratio at most $84$ under the conventions of Section~\ref{subsec:information-models}.
\end{theorem}

\begin{proof}
Finite-line sweeps ensure termination, and no positive phase means zero cost.
For the positive phases, the certificate groups are the anchor pairs $(p_{k-1},p_k)$, with windows
\begin{equation}
\label{eq:blind-service-display-10}U_k=(T_{k-2},T_k],
\end{equation}
for $k\ge3$, initially $[0,T_1]$ and $[0,T_2]$.
For $k\ge3$, the two requests satisfy
\begin{equation}
\label{eq:blind-service-display-11}a(p_{k-1})>T_{k-2}\qquad\text{and}\qquad T_{k-1}<a(p_k)\le T_k.
\end{equation}
Together with $a(p_{k-1})\le T_{k-1}$, this shows that each pair arrives in its window.
No request appears in two pairs of the same parity.
The first two windows contain their pairs directly, including the initial dummy.
Lemma~\ref{lem:blind-phase-bound} gives phase cost below $21R_k$ and
\begin{equation}
\label{eq:blind-service-display-12}T_k-T_{k-1}\ge r\ge R_k.
\end{equation}
Since the pair's Steiner cost is $R_k$, Corollary~\ref{cor:certificate-transfer} gives
\begin{equation}
\label{eq:blind-service-display-13}\ALG<21\sum_kR_k\le84\OPT,
\end{equation}
which proves the theorem.
\end{proof}

\subsection{Automatic service}
\label{subsec:blind-automatic-description}

Under automatic service, a search also serves every request it finds.
We keep records of those requests until the phase ends, preserving the search times and next anchor of the elective execution.

The algorithm creates a virtual request for each unresolved announcement and keeps that record until the phase ends, even after the real request is served.
A virtual request keeps its label, arrival time, revealed location if known, and discovered status.
The algorithm uses these records to keep the same oldest arrival time and search radii as the elective algorithm.
A phase ends only when every virtual request announced during the phase has been discovered.
Requests announced at the idle anchor are observed and served at age zero and are omitted from the numbered positive phases.
The same is true for arrivals at the server position during an active phase.
They need no virtual request because the elective reference algorithm also removes them for free.

During exploration at radius $r_j$, all real requests first discovered by the same instantaneous sweep are automatically served in one service event.
Their arrivals are no earlier than the oldest phase virtual request, implying that the maximum delay of this event is at most $r_j$.
A sweep that finds no unserved request has only movement cost.
Once every virtual request is discovered, every corresponding real request has already been served.
After the successful sweep returns to the anchor, the server moves directly to a farthest virtual request location within the same atomic action.
The algorithm then deletes all phase records and uses this point as the next anchor.

\begin{algorithm}[!htb]
 \caption{Blind automatic service with doubling search radii}
 \label{alg:blind-automatic-dyadic}
 \begin{algorithmic}[1]
    \Require Known finite line $X$ with $|X|\ge2$, initial anchor $s_0$
    \State $\delta\gets\min\{|x-y|:x,y\in X,\ x\ne y\}$; $s\gets s_0$
    \Loop
      \State Wait for the next complete arrival epoch and process it
      \State Resolve its automatic hits at $s$ and record the feedback from this check
      \State Create a virtual request for each identity unresolved after this check
      \If{no virtual request was created} \State \textbf{continue} \EndIf
      \State Let $\alpha$ denote the earliest virtual request arrival; $j\gets0$
      \Loop
        \State $r\gets2^j\delta$; $c\gets\alpha+r$; $u\gets$ next arrival time, or $+\infty$
        \If{$u\le c$}
          \State Advance to $u$, process its epoch, resolve anchor hits,
          and create virtual requests for the remaining identities
          \If{$u<c$} \State \textbf{continue} \EndIf
        \Else
          \State Advance to $c$
        \EndIf
        \State Sweep $s\to L_r\to R_r\to s$ as one automatic service component\label{line:blind-auto-sweep}
        \State Mark the virtual requests of all automatic hits as discovered
        \If{some virtual request is undiscovered}\label{line:blind-auto-failure}
          \State $j\gets j+1$; \textbf{continue}
        \EndIf
        \State Choose a farthest virtual request $p$ and move directly to $x(p)$\label{line:blind-auto-reposition}
        \Statex \Comment{The pure move shares the sweep's compound atomic action}
        \State $s\gets x(p)$; delete all phase virtual requests; \textbf{break}
      \EndLoop
    \EndLoop
  \end{algorithmic}
\end{algorithm}

\paragraph{Competitive ratio.}
\label{subsec:blind-automatic-analysis}

We sum the delay charges created by the successive discovery sweeps and bound the final repositioning cost.
The request pairs from Theorem~\ref{thm:blind-elective-competitive} then give the same competitive ratio.
\begin{theorem}
 \label{thm:blind-automatic-competitive}
On a known finite line with blind requests and automatic service, Algorithm~\ref{alg:blind-automatic-dyadic} is deterministic and has competitive ratio at most $84$ under the conventions of Section~\ref{subsec:information-models}.
\end{theorem}

\begin{proof}
The search test uses virtual requests (line~\ref{line:blind-auto-failure} in Algorithm~\ref{alg:blind-automatic-dyadic}), preserving exactly the same radii as in the elective algorithm.
The proof of Lemma~\ref{lem:blind-phase-bound} therefore gives termination and, for the final radius $r=2^j\delta$,
\begin{equation}
\label{eq:blind-service-display-14}R_k\le r\le2R_k.
\end{equation}
Each sweep in line~\ref{line:blind-auto-sweep}, at radius $2^i\delta$, costs at most $4\cdot2^i\delta$ in movement and $2^i\delta$ in delay.
The sum of sweep radii satisfies
\begin{equation}
\label{eq:blind-service-display-15}\sum_{i=0}^j2^i\delta<2r.
\end{equation}
Hence all sweeps cost less than $10r$.
Anchor hits are free, and the direct repositioning in line~\ref{line:blind-auto-reposition} costs $R_k$.
Thus
\begin{equation}
\label{eq:blind-service-display-16}\ALG_k<10r+R_k\le21R_k.
\end{equation}

Each anchor is still the location of an actual request from that phase.
Thus the same request pairs and windows satisfy the certificate theorem as in Theorem~\ref{thm:blind-elective-competitive}.
Using Theorem~\ref{thm:auto-elective-equivalence} to compare with the automatic optimum gives
\begin{equation}
\label{eq:blind-service-display-17}\sum_kR_k\le4\OPT_{\mathrm{auto}}.
\end{equation}
For a nonempty sequence of positive phases,
\begin{equation}
\label{eq:blind-service-display-18}\ALG_{\mathrm{auto}}<21\sum_kR_k\le84\OPT_{\mathrm{auto}}.
\end{equation}
Otherwise the cost is zero.
\end{proof}

\paragraph{Computational complexity.}
\label{subsec:blind-computation}

For rational inputs given in binary, both algorithms run in polynomial time.
After sorting the sites once, the algorithm represents each sweep and cleanup using $O(N)$ listed sites.
There are at most $Q$ positive phases, each using at most
\[
1+\left\lceil\log_2\frac{\operatorname{diam}(X)}{\delta}\right\rceil
\]
search radii. The total number of operations on listed sites is therefore
\[
O\!\left(QN\left(1+\log(\operatorname{diam}(X)/\delta)\right)\right).
\]
The logarithm is polynomially bounded by the binary input length.

\FloatBarrier
\section{Offline Algorithms and Lower Bounds}
\label{sec:offline-results}

The online bounds are complemented by exact offline algorithms and lower bounds.
Proposition~\ref{prop:elective-normal-form} reduces offline optimization to consecutive arrival blocks and their endpoints.
On lines and trees, the cost of serving one block can be computed explicitly, giving polynomial time dynamic programs.
On arbitrary metrics, even one arrival epoch yields an NP-hard routing problem.
We then give lower bounds for visible and blind service.

\paragraph{Notation.}
The requests have $m$ distinct arrival epochs $R_1,\ldots,R_m$ at times $a_1<\cdots<a_m$.
Let $Q$ denote the total number of request occurrences.
For a line, the sites are $X=\{x_1<\cdots<x_N\}$ with their usual distance.
The algorithms in Appendix~\ref{app:algorithms} store minimizing choices to reconstruct the schedules.

\subsection{Exact optimization on a line}
\label{sec:elective-offline}

By Proposition~\ref{prop:elective-normal-form}, it suffices to choose consecutive blocks of arrival epochs and a final server position for each block.
On a line, the shortest walk for a block must visit its leftmost and rightmost requests.
This gives the following dynamic program.

For $1\le p\le j\le m$, the block extremes are
\begin{equation}
\label{eq:elective-offline-display-1}L_{p,j}=\min\{x(q):q\in R_p\cup\cdots\cup R_j\},\qquad U_{p,j}=\max\{x(q):q\in R_p\cup\cdots\cup R_j\}.
\end{equation}
For a walk starting at $y\in X$ and ending at $z\in X$, the minimum length needed to visit the block is
\begin{equation}
\label{eq:line-route-kernel}K_{p,j}(y,z)=(U_{p,j}-L_{p,j})+\min\left\{|y-L_{p,j}|+|U_{p,j}-z|,|y-U_{p,j}|+|L_{p,j}-z|\right\}.
\end{equation}

Let $F[j,z]$ denote the minimum cost of a normalized schedule for epochs $1,\ldots,j$ ending at $z$.
Its initial values are
\begin{equation}
\label{eq:elective-offline-display-2}F[0,s_0]=0,\qquad F[0,z]=+\infty\quad(z\ne s_0).
\end{equation}
For each completed layer $h$, the algorithm computes
\begin{equation}
\label{eq:line-distance-transform}G_h(x)=\min_{y\in X}\{F[h,y]+|y-x|\},\qquad x\in X,
\end{equation}
which is the minimum cost of serving the first $h$ epochs and then moving to $x$.
It is also called the line distance transform of $F[h,\cdot]$.
A deterministically chosen minimizer $Y_h(x)$ is stored for reconstruction.

For $j\ge1$ and $z\in X$, the recurrence is
\begin{equation}
\label{eq:elective-offline-dp}F[j,z]=\min_{\substack{1\le p\le j\\
y\in X}}\left\{F[p-1,y]+K_{p,j}(y,z)+a_j-a_p\right\}.
\end{equation}

Algorithm~\ref{alg:line-offline-dp} evaluates the recurrence stated in \eqref{eq:elective-offline-dp} and records the choices needed to reconstruct a schedule.
All ties are resolved deterministically.

\paragraph{Correctness and running time.}
\label{subsec:line-offline-analysis}

We first prove that the recurrence computes the optimum.
The values $G_h$ allow minimization over the preceding endpoint in constant time per transition.
\begin{theorem}
 \label{thm:elective-offline-dp}
Under elective service on a finite line, Algorithm~\ref{alg:line-offline-dp} evaluates recurrence~\eqref{eq:elective-offline-dp} and satisfies
 \begin{equation}
\label{eq:elective-offline-display-3}\OPT=\min_{z\in X}F[m,z].
\end{equation}
If $Q$ is the number of request occurrences, the implementation using the distance transform runs in $O(Q+m^2N)$ time and uses $O(mN)$ working space beyond the input, including traceback pointers.
\end{theorem}

\begin{proof}
Any walk serving a block must visit both of its extreme locations.
It visits either the left extreme or the right extreme first, giving the two expressions in~\eqref{eq:line-route-kernel}.
Following the shorter of these two walks attains their minimum, including when the extremes coincide.
Thus a final block $[p,j]$ from $y$ to $z$ costs exactly $K_{p,j}(y,z)+a_j-a_p$.
Choosing the first epoch $p$ of the last block and its starting point $y$ accounts for every normalized schedule.
Induction on $j$ therefore proves that the recurrence finds the best such schedule, which is optimal by Proposition~\ref{prop:elective-normal-form}.
The stored choices reconstruct a schedule of this cost.

The distance transform evaluates a transition without scanning predecessor endpoints, since
\begin{align}
&\min_{y\in X}\{F[p-1,y]+K_{p,j}(y,z)\}\notag\\
&\quad=(U_{p,j}-L_{p,j})+\min\left\{G_{p-1}(L_{p,j})+|U_{p,j}-z|,G_{p-1}(U_{p,j})+|L_{p,j}-z|\right\}.\label{eq:line-dp-transformed-transition}
\end{align}
Each layer's transform is the minimum of the forward scan
\begin{equation}
\label{eq:elective-offline-display-4}G_h^{\leftarrow}(x_i)=\min\{F[h,x_i],G_h^{\leftarrow}(x_{i-1})+x_i-x_{i-1}\},
\end{equation}
initialized at $x_1$, and its backward analogue.
Both scans take $O(N)$ time and also record the minimizing endpoints.
After $O(Q)$ preprocessing of epoch extremes, scanning $p$ backwards for each $j$ maintains block extremes in constant time and evaluates all $O(m^2N)$ candidates.
The tables, transform minimizers, and traceback choices occupy $O(mN)$ space.
\end{proof}

The endpoint index in \eqref{eq:elective-offline-dp} records the state needed by the next batch.
Theorem~\ref{thm:auto-elective-equivalence} then transfers the computed value to automatic service.
\begin{corollary}
 \label{cor:auto-offline-dp}
On a finite line, recurrence~\eqref{eq:elective-offline-dp} computes the exact automatic service optimum in $O(Q+m^2N)$ time and $O(mN)$ working space beyond the input, where $Q$ is the number of request occurrences.
An optimal automatic schedule follows the walks reconstructed by the dynamic program, as in Theorem~\ref{thm:auto-elective-equivalence}.
\end{corollary}

\subsection{Exact optimization on a weighted tree}
\label{subsec:tree-offline-description}

For a finite weighted tree $T=(V,E)$, the dynamic program again chooses consecutive arrival blocks and their endpoints.
The notation $T[A]$ denotes the minimal subtree spanning $A$, and $w$ denotes its total edge weight.
We first compute the shortest walk for a block with a given start and end.

For an arrival block $[p,j]$, let $H_{p,j}$ denote the minimal subtree spanning the request locations in epochs $p,\ldots,j$, and write
\begin{equation}
\label{eq:tree-extension-display-4}\eta_{p,j}(v)=d(v,H_{p,j})\qquad(v\in V).
\end{equation}
For vertices $y,z$, let $T_{p,j}(y,z)$ denote the minimal subtree containing $H_{p,j}\cup\{y,z\}$.
The shortest walk length follows by counting how often each edge must be crossed.
\begin{lemma}
 \label{lem:tree-route-kernel}
The minimum length of a walk starting at $y$, visiting every request in epochs $p,\ldots,j$, and ending at $z$ is
 \begin{equation}
\label{eq:tree-route-kernel}K^T_{p,j}(y,z)=2w(T_{p,j}(y,z))-d(y,z).
\end{equation}
Equivalently, writing $a=\eta_{p,j}(y)$, $b=\eta_{p,j}(z)$, and $c=\min\{d(y,z),a+b\}$,
 \begin{equation}
\label{eq:tree-route-kernel-fast}K^T_{p,j}(y,z)=2w(H_{p,j})+a+b+c-d(y,z).
\end{equation}
\end{lemma}

\begin{proof}
Each edge of $T_{p,j}(y,z)$ separates required vertices.
Edges on the $y$--$z$ path must be crossed at least once and all other edges at least twice.
An open depth first traversal attains these multiplicities, proving~\eqref{eq:tree-route-kernel}.

After contraction of $H_{p,j}$ to a root, we consider its distances to $y,z$.
The distances among the root and endpoints $y,z$ are $a,b,c$, since contraction changes $d(y,z)$ to $\min\{d(y,z),a+b\}$.
Each edge of their minimal connecting tree lies on exactly two of the three pairwise paths.
The weight added to $H_{p,j}$ is therefore $(a+b+c)/2$.
Substitution gives~\eqref{eq:tree-route-kernel-fast}.
The formula also holds when an endpoint lies in $H_{p,j}$ or the endpoints coincide.
\end{proof}

Let $F^T[j,z]$ denote the minimum cost of a normalized schedule for epochs $1,\ldots,j$ ending at $z$.
The recurrence is
\begin{equation}
\label{eq:tree-offline-dp}F^T[j,z]=\min_{\substack{1\le p\le j\\
y\in V}}\left\{F^T[p-1,y]+2w(T_{p,j}(y,z))-d(y,z)+a_j-a_p\right\},
\end{equation}
with $F^T[0,s_0]=0$ and all other states in layer zero infinite.

\paragraph{Correctness and running time.}
\label{subsec:tree-offline-analysis}

We prove that recurrence~\eqref{eq:tree-offline-dp} attains the offline optimum and bound its running time by counting block and endpoint combinations.
Theorem~\ref{thm:auto-elective-equivalence} also gives an automatic schedule of this cost.
\begin{theorem}
 \label{thm:tree-offline-dp}
Under elective service on a finite weighted tree, Algorithm~\ref{alg:tree-offline-dp} evaluates recurrence~\eqref{eq:tree-offline-dp} and satisfies
 \begin{equation}
\label{eq:tree-extension-display-7}\OPT=\min_{z\in V}F^T[m,z].
\end{equation}
If $n=|V|$ and $Q$ is the number of request occurrences, a direct implementation runs in $O(Q+mn+m^2n^2)$ time.
The algorithm also gives the automatic optimum and a schedule attaining it.
\end{theorem}

\begin{proof}
By Proposition~\ref{prop:elective-normal-form}, an optimal elective schedule consists of consecutive arrival blocks.
If its last block is $[p,j]$ and starts at $y$, Lemma~\ref{lem:tree-route-kernel} gives exactly the movement term in \eqref{eq:tree-offline-dp}; its delay is $a_j-a_p$.
Minimizing over $p$ and $y$ proves the recurrence, and the stored choices reconstruct an optimal schedule.
Following these walks under automatic service is optimal by Theorem~\ref{thm:auto-elective-equivalence}.

Vertex prefix counts cost $O(Q+mn)$ and all tree distances cost $O(n^2)$ to compute.
For each of $O(m^2)$ blocks, the counts and tree traversals identify its terminals, subtree weight, and distances $\eta_{p,j}$ in $O(n)$ time.
Formula~\eqref{eq:tree-route-kernel-fast} evaluates the $n^2$ endpoint pairs in constant time each.
Stored minimizing choices recover the blocks and traversals within the resulting $O(Q+mn+m^2n^2)$ bound.
\end{proof}

\subsection{Hardness on general metrics}
\label{subsec:metric-offline-complexity}

We prove NP hardness of exact offline optimization by reducing Hamiltonian Path to an instance with a single arrival epoch.
Even with no waiting cost, a service walk must solve the problem of visiting all requested locations from a prescribed start.
\begin{proposition}
 \label{prop:general-metric-offline-hardness}
Exact offline optimization on an arbitrary finite metric is NP hard under both elective and automatic service, even when all requests arrive at time zero.
\end{proposition}

\begin{proof}
Given a Hamiltonian Path instance $G=(V,E)$ with $n$ vertices, we add a universal vertex $s$ and use its shortest-path metric
\begin{equation}
\label{eq:general-metric-display-16}d(s,v)=1\quad(v\in V),\qquad d(u,v)=\begin{cases}1,&uv\in E,\\
2,&uv\notin E\end{cases}\quad(u,v\in V,\ u\ne v).
\end{equation}
The matrix is computable in polynomial time.
The server starts at $s$, and one request arrives at each vertex of $V$ at time zero.
Under either rule, $\OPT$ equals the shortest visiting-walk length $L^\star$, since such a walk can be executed at time zero without delay.
Shortcutting any walk to its first-visit order gives
\begin{equation}
\label{eq:general-metric-display-17}\len(P)\ge d(s,v_{\pi(1)})+\sum_{i=1}^{n-1}d(v_{\pi(i)},v_{\pi(i+1)}).
\end{equation}
Every permutation is also feasible, implying
\begin{equation}
\label{eq:general-metric-display-18}L^\star=n+\min_{\pi}\bigl|\{i\in[n-1]:v_{\pi(i)}v_{\pi(i+1)}\notin E\}\bigr|.
\end{equation}
Thus $\OPT\le n$ if and only if $G$ has a Hamiltonian path.
\end{proof}

Thus exact offline optimization admits polynomial running time on finite lines and explicitly represented weighted trees, but is NP hard on arbitrary finite metrics.

\subsection{Deterministic lower bound}
\label{sec:visible-lower-bound}

We prove a lower bound by repeatedly releasing a request opposite the online server.
We construct three offline schedules whose total cost is at most the online cost plus an error that is small compared with it.
At least one of them then costs about a third as much as the online algorithm.
The construction applies to a line with two sites and therefore to every metric class considered in this paper.
\begin{theorem}
 \label{thm:visible-two-point-lower-bound}
Let $X=\{0,D\}$ denote a line with two sites, where $D>0$, and let $0$ denote the initial server position.
Under either elective or automatic service, every deterministic visible online algorithm has competitive ratio at least $3$.
More precisely, for every $\rho<3$ there is a finite input $I$ with $\ALG(I)>\rho\OPT(I)$.
The conclusion also survives a fixed additive constant in the definition of competitiveness.
\end{theorem}

\begin{proof}
Let $n\ge1$ denote the number of requests and let $\varepsilon>0$ denote a time separation.
The first request $q_1$ arrives at $D$ at time zero.
Let $t_k$ denote the service time of $q_k$ and $w_k$ its waiting time:
\begin{equation}
\label{eq:visible-lower-bounds-display-1}w_k=t_k-a_k.
\end{equation}
For $k<n$, the adversary releases $q_{k+1}$ at $a_{k+1}=t_k+\varepsilon$ opposite the server just before that epoch.
A request never served makes the algorithm infeasible.
Writing $x_k$ for the request location, we define
\begin{equation}
\label{eq:visible-lower-bounds-display-2}e_k=\mathbf{1}\{x_{k+1}=x_k\},\qquad E=\sum_{k=1}^{n-1}e_k.
\end{equation}
Between consecutive releases, visiting $x_k$ and ending opposite $x_{k+1}$ requires movement at least $(1+e_k)D$.
The final request needs another $D$.
Distinct service events therefore give
\begin{equation}
\label{eq:visible-lower-alg}\ALG\ge (n+E)D+\sum_{k=1}^n w_k.
\end{equation}

We inductively maintain three offline schedules.
After epoch $k$, $P$ waits opposite $x_k$ with only $q_k$ pending, while schedules $A$ and $B$ have served all requests and stand opposite $x_k$ and at $x_k$, respectively.
Initially, the three schedules use walks $0$, $0\to D\to0$, and $0\to D$, costing $3D$ in total.
At $a_{k+1}$, all arrivals at parked points are served for free, and schedule $P$ moves to $x_k$ to clear its pending requests.
If $x_{k+1}=x_k$, schedule $B$ also moves to the opposite point.
If $x_{k+1}\ne x_k$, schedule $B$ now has only $q_{k+1}$ pending, and the old $P$ and $A$ have served everything.
If $x_{k+1}=x_k$, the old $A$ again has only $q_{k+1}$ pending, while $P$ has served both requests and $B$ has moved after serving the new arrival.
In either case, the three schedules can be renamed according to their positions and pending requests to obtain the same states for epoch $k+1$.
This update costs $(1+e_k)D$ in movement and $a_{k+1}-a_k$ in delay under either service rule.
At the final arrival time $a_n$, schedule $P$ serves $q_n$ by moving distance $D$; this adds no delay.

The three costs $C_1,C_2,C_3$ total $(n+E+3)D$ in movement and $\sum_{k<n}(a_{k+1}-a_k)$ in delay.
Since
\begin{equation}
\label{eq:visible-lower-bounds-display-3}a_{k+1}-a_k=w_k+\varepsilon,
\end{equation}
we obtain
\begin{equation}
\label{eq:visible-lower-opt}3\OPT\le C_1+C_2+C_3\le (n+E+3)D+\sum_{k=1}^{n-1}w_k+(n-1)\varepsilon\le \ALG+3D+(n-1)\varepsilon.
\end{equation}
Consequently,
\begin{equation}
\label{eq:visible-lower-bounds-display-4}\frac{\ALG}{\OPT}\ge\frac{3\ALG}{\ALG+3D+(n-1)\varepsilon}\ge\frac{3nD}{(n+3)D+(n-1)\varepsilon}.
\end{equation}
Taking $\varepsilon=D/n^2$ makes this tend to $3$ as $n\to\infty$.
For $1\le\rho<3$, the same bounds give
\begin{equation}
\label{eq:visible-lower-bounds-display-5}\ALG-\rho\OPT\ge\left(\frac{(3-\rho)n}{3}-\rho\right)D-\frac{\rho(n-1)\varepsilon}{3}.
\end{equation}
The same choice of $\varepsilon$ makes the right side tend to infinity, ruling out any fixed additive constant.
For a deterministic algorithm the recursion can be simulated in advance, implying that each construction is a fixed finite input.
\end{proof}

The same instance with two sites is itself a line and a weighted tree, and is a valid arbitrary finite metric.
Theorem \ref{thm:visible-two-point-lower-bound} therefore provides a lower bound for every visible setting in this paper.

\subsection{Lower bounds for hidden locations}
\label{subsec:blind-lower-bounds}

\paragraph{Randomized guarantees.}
A randomized algorithm is feasible if every request is eventually served with probability one on every finite input, and its competitive ratio is at most $\rho$ if
\begin{equation}
\label{eq:model-display-5}\mathbb{E}_{R}\!\left[\ALG_R(I)\right]\le \rho\OPT(I)
\end{equation}
for every finite input $I$, where the expectation is over the algorithm's internal random seed $R$.
Unless explicitly stated otherwise, the adversary is oblivious, meaning that it fixes $I$ independently of $R$. No additive constant is included in either definition.

Even when exactly one request arrives, the server may have to search several locations before finding it.
Counting the movement needed to reach each possible location gives the following lower bounds.
\begin{theorem}
 \label{thm:blind-search-lower-bounds}
Suppose the algorithm is told in advance that exactly one hidden request arrives at time zero and that no request will arrive later.
 \begin{enumerate}
 \item On the line with three sites $\{-1,0,1\}$, with the server initially at $0$, every deterministic algorithm has competitive ratio at least $3$, and every randomized algorithm has ratio at least $2$ against an oblivious adversary.
 \item On a unit star with $d$ leaves, with the server initially at its center, the corresponding lower bounds are $2d-1$ deterministically and $d$ randomized.
 \end{enumerate}
Both statements hold under elective and automatic service.
\end{theorem}

\begin{proof}
Until the request is found, each visited leaf only tells the algorithm that the request is not there.
For a fixed choice of random bits, these unsuccessful visits determine an order of searching the leaves.
Reaching the $k$th distinct leaf requires movement at least $1+2(k-1)$.
A deterministic search therefore pays at least $2d-1$ for the last leaf, whereas $\OPT=1$.
For a randomized search, we consider a choice of random bits for which the algorithm eventually finds the request at every leaf.
Feasibility and the finite number of leaves imply that this holds with probability one.
A uniform hidden leaf gives average movement at least
\begin{equation}
\label{eq:blind-service-display-1}\frac1d\sum_{k=1}^d(2k-1)=d.
\end{equation}
Averaging also over the random bits shows that some fixed leaf has expected search cost at least $d$.
Waiting and deferring service cannot decrease discovery movement, implying that both rules satisfy the bounds.
The three-site line is the case $d=2$.
\end{proof}

\FloatBarrier
\section{Conclusion}
\label{sec:conclusion}

We studied online service with maximum delay per batch.
The certificate bound compares groups of requests with any offline schedule, even when that schedule serves different batches and ends its walks at different locations.
This bound gives constant competitive ratios on finite lines, trees, and general metrics, and also for blind service on a known finite line.
The same offline value applies to elective and automatic service, although their online implementations can differ.

The difference from additive delay is that younger requests can join a
batch without increasing its delay charge.
For visible requests, a common waiting clock suffices to coordinate service
across the metric, while the request locations determine the threshold and
walk. The certificate bound justifies serving requests together at only a
constant loss, despite the offline optimum's freedom to choose its batches.

The competitive ratios leave several questions open.
The deterministic lower bound is $3$ for visible service, compared with upper bounds of $10$ on lines, $12$ on trees, and $20$ on general metrics in polynomial time.
For blind service on a line, the deterministic bounds are $3$ and $84$, and the randomized lower bound is $2$.
Closing these gaps, and determining whether randomization helps for requests arriving over time, remain open.
The star lower bound also raises the question of which other metrics permit constant competitive ratios with hidden locations.

Offline optimization is polynomial on lines and trees and NP-hard on general metrics.
An approximation algorithm for the latter case must choose both the batches and their walks, since one walk's endpoint is the next walk's starting point.
For online service, Proposition~\ref{prop:spatial-surrogate-transfer} converts a better connecting-tree approximation into a better competitive ratio.
It remains open whether choosing the waiting threshold and the service walk together can give a stronger guarantee.

\appendix
\section{Conversion from elective to automatic service}
\label{app:automatic-replay}
\setcounter{algorithm}{0}
\renewcommand{\thealgorithm}{\thesection.\arabic{algorithm}}
\renewcommand{\theHalgorithm}{\thesection.\arabic{algorithm}}

Algorithm~\ref{alg:automatic-replay} describes this conversion.
It uses the consecutive arrival blocks of the offline schedule.
An arbitrary elective online algorithm may instead visit an old request that it plans to serve in a later batch; therefore, the conversion does not apply to every online algorithm.

\begin{algorithm}[ht]
 \caption{Automatic replay of a normalized elective schedule}
 \label{alg:automatic-replay}
 \begin{algorithmic}[1]
    \Require Consecutive elective blocks $B_1,\ldots,B_K$, service times
      $t_k=\beta(B_k)$, and walks with matching consecutive endpoints $P_1,\ldots,P_K$
    \State $k\gets1$
    \For{arrival epochs $h=1,\ldots,m$ in chronological order}
      \State Park at the current endpoint and process the complete epoch $R_h$
      \State Automatically serve all arrivals in $R_h$ at the parked point
        and delete them from their designated blocks
      \If{$a_h=t_k$}
        \State Traverse all of $P_k$ after processing $R_h$
        \State Automatically serve every remaining request of $B_k$ hit by $P_k$
        \State Retain $P_k$ as pure repositioning if $B_k$ was already empty
        \State $k\gets k+1$
      \EndIf
    \EndFor
    \State \Return the automatic schedule
  \end{algorithmic}
\end{algorithm}

Requests served at arrival at the parked server may split a planned elective batch into several actual batches.
Those early services are free, which is why the total cost does not increase.

\clearpage
\section{Pseudocode for Offline Algorithms}
\label{app:algorithms}
\setcounter{algorithm}{0}
\renewcommand{\thealgorithm}{\thesection.\arabic{algorithm}}
\renewcommand{\theHalgorithm}{\thesection.\arabic{algorithm}}

The algorithms below implement the offline recurrences of Section~\ref{sec:offline-results}.
All ties are resolved by fixed orders on sites, edges, and request labels.

\begin{algorithm}[H]
 \caption{Exact offline dynamic program on a finite line}
 \label{alg:line-offline-dp}
 \begin{algorithmic}[1]
    \Require Arrival epochs $(R_1,a_1),\ldots,(R_m,a_m)$, sites $X$, start $s_0$
    \State Precompute the epoch extremes $L_{p,p},U_{p,p}$ for $p=1,\ldots,m$
    \State Set $F[0,s_0]\gets0$ and $F[0,z]\gets+\infty$ for $z\ne s_0$
    \State Compute $G_0(x)$ and minimizers $Y_0(x)$ by forward/backward scans
    \For{$j=1,\ldots,m$}
      \State Set $F[j,z]\gets+\infty$ for every $z\in X$
      \State $L\gets+\infty$; $U\gets-\infty$
      \For{$p=j,j-1,\ldots,1$}
        \State $L\gets\min\{L,L_{p,p}\}$; $U\gets\max\{U,U_{p,p}\}$
        \For{$z\in X$}
          \State $v_L\gets G_{p-1}(L)+|U-z|$;
            $v_U\gets G_{p-1}(U)+|L-z|$
          \State Choose the smaller term using the fixed tie order
          \State Set $y\gets Y_{p-1}(L)$ if $v_L$ wins, and
            $y\gets Y_{p-1}(U)$ otherwise
          \State $v\gets(U-L)+(a_j-a_p)+\min\{v_L,v_U\}$
          \If{$v<F[j,z]$, or the candidate wins the fixed rule for resolving ties}
            \State $F[j,z]\gets v$;
              $\operatorname{pred}[j,z]\gets(p,y,L,U,\text{winning order})$
          \EndIf
        \EndFor
      \EndFor
      \State Compute $G_j(x)$ and minimizers $Y_j(x)$ by two line scans
    \EndFor
    \State $z_m\gets\arg\min_{z\in X}F[m,z]$
    \State Trace back $\operatorname{pred}$; for every recovered block $[p,j]$,
      execute a shortest walk attaining $K_{p,j}(y,z)$ at time $a_j$ and electively serve
      exactly the identities in epochs $p,\ldots,j$
    \State \Return the resulting schedule and $F[m,z_m]$
  \end{algorithmic}
\end{algorithm}
\FloatBarrier

\begin{algorithm}[H]
 \caption{Exact offline dynamic program on a weighted tree}
 \label{alg:tree-offline-dp}
 \begin{algorithmic}[1]
    \Require Weighted tree $T=(V,E)$, initial vertex $s_0$, arrival epochs
    \State Precompute tree distances and request prefix counts for each vertex
    \State Set $F^T[0,s_0]\gets0$ and all other states in layer zero to $+\infty$
    \For{$j=1,\ldots,m$}
      \State Set $F^T[j,z]\gets+\infty$ for every $z\in V$
      \For{$p=j,j-1,\ldots,1$}
        \State Construct $H_{p,j}$, its weight, and all
          $\eta_{p,j}(v)=d(v,H_{p,j})$
        \For{$y,z\in V$}
          \State Evaluate $K^T_{p,j}(y,z)$ in $O(1)$ from
            \eqref{eq:tree-route-kernel-fast}
          \State Relax~\eqref{eq:tree-offline-dp} and store $(p,y)$ using the fixed tie order
        \EndFor
      \EndFor
    \EndFor
    \State Backtrack the consecutive epoch blocks and their endpoints
    \For{each recovered block $[p,j]$ from endpoint $y$ to endpoint $z$}
      \State At time $a_j$, use an open depth first traversal of $T_{p,j}(y,z)$ from $y$ to $z$
      \Statex \Comment{Use each edge outside the path from $y$ to $z$ twice, and each path edge once}
      \State Electively serve exactly the identities in epochs $p,\ldots,j$
    \EndFor
    \State \Return the reconstructed schedule and $\min_zF^T[m,z]$
  \end{algorithmic}
\end{algorithm}
\FloatBarrier

\clearpage
\section{Additional Examples and Search Bounds}
\label{app:supplementary}

\subsection{Two-request groups and a sharp certificate bound}
The next lemma specializes Theorem~\ref{thm:certificate-engine} to pairs of requests.
The example following its proof shows that the coefficient two is sharp.
\begin{lemma}
 \label{lem:single-color-matching}
The times
 \begin{equation}
\label{eq:online-global-hull-display-4}0=T_0<T_1<\cdots<T_m
\end{equation}
define ordered windows with disjoint interiors $W_1=[T_0,T_1]$ and $W_k=(T_{k-1},T_k]$ for $k\ge2$.
Suppose that the certificate instance contains exactly two requests in $W_k$, their locations are at distance $X_k$, and
 \begin{equation}
\label{eq:window-dominates-span}T_k-T_{k-1}\ge X_k.
\end{equation}
Then every feasible elective schedule $\sigma$ satisfies
 \begin{equation}
\label{eq:single-color-bound}\sum_{k=1}^m X_k\le\operatorname{Delay}(\sigma)+2\operatorname{Mov}(\sigma)\le 2\cost(\sigma).
\end{equation}
Consequently, $\sum_kX_k\le2\OPT$.
\end{lemma}

\begin{proof}
Theorem~\ref{thm:certificate-engine} applies to the designated pairs with $C_k=X_k$, since the Steiner cost of two points is their distance.
\end{proof}

The coefficient two in Lemma~\ref{lem:single-color-matching} is best possible.
The example uses two windows of unit length on the line with sites $\{0,1\}$.
One endpoint pair arrives at time $1$ and the second at time $1+\varepsilon$, where $0<\varepsilon\le1$.
A schedule starting at zero serves the first request there for free and at time $1+\varepsilon$ traverses from zero to one, serving the remaining three requests.
The cost is $1+\varepsilon$, whereas the two certificates sum to two.
Taking $\varepsilon\downarrow0$ shows that no smaller coefficient works uniformly.

\subsection{Search for a single request}
\label{subsec:blind-single-search}

We show that alternating excursions find one hidden request on a finite line with movement bounded by its distance from the starting point.
If $|X|=1$, every request is at the known server point and all blind costs are zero.
The following argument assumes $|X|\ge2$.
\begin{proposition}
 \label{prop:blind-one-request}
Let $X\subset\mathbb{R}$ denote a finite set and let $s\in X$ denote the initial server position.
The minimum site spacing is
 \begin{equation}
\label{eq:blind-service-display-2}\delta=\min\{|x-y|:x,y\in X,\ x\ne y\}>0.
\end{equation}
A hidden request at distance $d>0$ can be found deterministically with movement less than $9d$.
A request at distance zero is found with zero movement.
\end{proposition}

\begin{proof}
The server first checks $s$ and then alternates left and right excursions of radii
\begin{equation}
\label{eq:blind-service-display-3}\delta,2\delta,4\delta,8\delta,\ldots,
\end{equation}
visiting every site reached and returning to $s$ after failure, turning at an endpoint if necessary.
If discovery occurs at radius $L$, then $d>L/4$.
For the first two excursions this follows from $d\ge\delta$.
Thereafter it follows from failure of the preceding same-side excursion of radius $L/4$.
All completed excursions cost less than $2L$, and the final segment costs $d$.
Thus the total is less than
\begin{equation}
\label{eq:blind-service-display-4}2L+d<9d,
\end{equation}
which proves the proposition.
\end{proof}

\clearpage
\section*{Disclosure of AI Assistance in Research and Writing}
The authors used language models from OpenAI, including Codex, extensively and substantively as interactive assistants during this project.
These systems contributed to proof exploration, counterexample search, formalization, adversarial checking, literature organization, LaTeX drafting, and editorial revision.
This assistance materially affected the technical development and proofs in Sections~\ref{sec:certificate-framework}--\ref{sec:offline-results} and Appendices~\ref{app:automatic-replay}--\ref{app:supplementary}, as well as the exposition throughout the paper.
The human authors selected the research questions and modeling choices, evaluated and verified the generated arguments, and take full responsibility for every claim and for the final manuscript.

\bibliographystyle{plain}
\bibliography{reference}

@article{mitrovic2004waiting,
  author  = {Mitrovi{\'c}-Mini{\'c}, Sne{\v z}ana and Laporte, Gilbert},
  title   = {Waiting strategies for the dynamic pickup and delivery problem with time windows},
  journal = {Transportation Research Part B: Methodological},
  volume  = {38},
  number  = {7},
  pages   = {635--655},
  year    = {2004},
  doi     = {10.1016/j.trb.2003.09.002}
}

@article{branke2005waiting,
  author  = {Branke, J{\"u}rgen and Middendorf, Martin and Noeth, Guntram and Dessouky, Maged},
  title   = {Waiting Strategies for Dynamic Vehicle Routing},
  journal = {Transportation Science},
  volume  = {39},
  number  = {3},
  pages   = {298--312},
  year    = {2005},
  doi     = {10.1287/trsc.1040.0095}
}

@article{dooly2001tcp,
  author  = {Dooly, Daniel R. and Goldman, Sally A. and Scott, Stephen D.},
  title   = {On-line analysis of the {TCP} acknowledgment delay problem},
  journal = {Journal of the ACM},
  volume  = {48},
  number  = {2},
  pages   = {243--273},
  year    = {2001}
}

@inproceedings{azar2019framework,
  title={General framework for metric optimization problems with delay or with deadlines},
  author={Azar, Yossi and Touitou, Noam},
  booktitle={Proc. FOCS},
  pages={60--71},
  year={2019},
  organization={IEEE}
}

@inproceedings{karlin2001dynamic,
  title={Dynamic TCP acknowledgement and other stories about e/(e-1)},
  author={Karlin, Anna R. and Kenyon, Claire and Randall, Dana},
  booktitle={Proc. STOC},
  pages={502--509},
  year={2001}
}

@inproceedings{seiden2000guessing,
  title={A guessing game and randomized online algorithms},
  author={Seiden, Steven S.},
  booktitle={Proc. STOC},
  pages={592--601},
  year={2000}
}

@inproceedings{bhore2026online,
  title={Online TCP acknowledgment under general delays},
  author={Bhore, Sujoy and Paw{\l}owski, Micha{\l} and Umboh, Seeun William},
  booktitle={Proc. APPROX/RANDOM},
  pages={26:1--26:24},
  year={2026}
}

@article{albers2005dynamic,
  title={Dynamic {TCP} acknowledgment: Penalizing long delays},
  author={Albers, Susanne and Bals, Helge},
  journal={SIAM Journal on Discrete Mathematics},
  volume={19},
  number={4},
  pages={938--951},
  year={2005}
}

@inproceedings{azar2017osd,
  author    = {Azar, Yossi and Ganesh, Arun and Ge, Rong and Panigrahi, Debmalya},
  title     = {Online Service with Delay},
  booktitle = {Proc. STOC},
  pages     = {551--563},
  year      = {2017},
  doi       = {10.1145/3055399.3055475}
}

@inproceedings{bienkowski2018osdline,
  author    = {Bie{\'n}kowski, Marcin and Kraska, Artur and Schmidt, Pawe{\l}},
  title     = {Online Service with Delay on a Line},
  booktitle = {Proc. SIROCCO},
  series    = {LNCS},
  volume    = {11085},
  pages     = {237--248},
  year      = {2018},
  doi       = {10.1007/978-3-030-01325-7_22}
}

@inproceedings{touitou2023improved,
  author    = {Touitou, Noam},
  title     = {Improved and Deterministic Online Service with Deadlines or Delay},
  booktitle = {Proc. STOC},
  pages     = {761--774},
  year      = {2023},
  doi       = {10.1145/3564246.3585107}
}

@inproceedings{azar2020beyond,
  author    = {Azar, Yossi and Touitou, Noam},
  title     = {Beyond Tree Embeddings---A Deterministic Framework for Network Design with Deadlines or Delay},
  booktitle = {Proc. FOCS},
  year      = {2020},
  eprint    = {2004.07946},
  archivePrefix = {arXiv}
}

@misc{lu2026mlamax,
  author       = {Lu, Tianhang and Ren, Runtian and Liu, Shengcai and Tang, Ke},
  title        = {Online Multi-Level Aggregation with Maximum Delay},
  year         = {2026},
  eprint       = {2608.06796},
  archivePrefix = {arXiv},
  primaryClass = {cs.DS},
  note         = {arXiv:2608.06796},
  url          = {https://arxiv.org/abs/2608.06796}
}

@inproceedings{disser2025lazy,
  author    = {Disser, Yann and Thelen, Linda},
  title     = {A Tight Lower Bound for Online Service with Deadlines and Lazy Server},
  booktitle = {Proc. ISAAC},
  series    = {LIPIcs},
  volume    = {359},
  pages     = {26:1--26:17},
  year      = {2025},
  doi       = {10.4230/LIPIcs.ISAAC.2025.26}
}

@inproceedings{touitou2025nonclairvoyant,
  author    = {Touitou, Noam},
  title     = {Nearly-Optimal Algorithm for Non-Clairvoyant Service with Delay},
  booktitle = {Proc. STACS},
  series    = {LIPIcs},
  volume    = {327},
  pages     = {74:1--74:21},
  year      = {2025},
  doi       = {10.4230/LIPIcs.STACS.2025.74}
}

@inproceedings{emek2016matching,
  author    = {Emek, Yuval and Kutten, Shay and Wattenhofer, Roger},
  title     = {Online Matching: Haste Makes Waste!},
  booktitle = {Proc. STOC},
  pages     = {333--344},
  year      = {2016},
  doi       = {10.1145/2897518.2897557}
}

@inproceedings{azar2017matching,
  author    = {Azar, Yossi and Chiplunkar, Ashish and Kaplan, Haim},
  title     = {Polylogarithmic Bounds on the Competitiveness of Min-Cost Perfect Matching with Delays},
  booktitle = {Proc. SODA},
  pages     = {1051--1061},
  year      = {2017},
  doi       = {10.1137/1.9781611974782.67}
}

@inproceedings{ashlagi2017bipartite,
  author    = {Ashlagi, Itai and Azar, Yossi and Charikar, Moses and Chiplunkar, Ashish and Geri, Ofir and Kaplan, Haim and Makhijani, Rahul and Wang, Yuyi and Wattenhofer, Roger},
  title     = {Min-Cost Bipartite Perfect Matching with Delays},
  booktitle = {Proc. APPROX/RANDOM},
  series    = {LIPIcs},
  volume    = {81},
  pages     = {1:1--1:20},
  year      = {2017},
  doi       = {10.4230/LIPIcs.APPROX-RANDOM.2017.1}
}

@inproceedings{azar2021concave,
  author    = {Azar, Yossi and Ren, Runtian and Vainstein, Danny},
  title     = {The Min-Cost Matching with Concave Delays Problem},
  booktitle = {Proc. SODA},
  pages     = {301--320},
  year      = {2021},
  doi       = {10.1137/1.9781611976465.20}
}

@inproceedings{dufay2026matching,
  author    = {Dufay, Marc and Wattenhofer, Roger},
  title     = {A Deterministic Polylogarithmic Competitive Algorithm for Matching with Delays},
  booktitle = {Proc. SODA},
  pages     = {3936--3964},
  year      = {2026},
  doi       = {10.1137/1.9781611978971.144}
}

@article{gupta2022caching,
  author  = {Gupta, Anupam and Kumar, Amit and Panigrahi, Debmalya},
  title   = {Caching with Time Windows and Delays},
  journal = {SIAM Journal on Computing},
  volume  = {51},
  number  = {4},
  pages   = {975--1017},
  year    = {2022},
  doi     = {10.1137/20M1346286}
}

@inproceedings{krnetic2020kserver,
  author    = {Krneti\'c, Predrag and Melnyk, Darya and Wang, Yuyi and Wattenhofer, Roger},
  title     = {The $k$-Server Problem with Delays on the Uniform Metric Space},
  booktitle = {Proc. ISAAC},
  series    = {LIPIcs},
  volume    = {181},
  pages     = {61:1--61:13},
  year      = {2020},
  doi       = {10.4230/LIPIcs.ISAAC.2020.61}
}

\end{document}